\documentclass[twocolumn,showpacs,preprintnumbers, nofootinbib, superscriptaddress, amsmath,amssymb]{revtex4}
\usepackage{amsthm}
\usepackage{bm}
\usepackage{epsfig}

\DeclareFontFamily{OT1}{rsfs}{}
\DeclareFontShape{OT1}{rsfs}{m}{n}{<5> rsfs5 <7> rsfs7 <10> rsfs10}{}
\DeclareSymbolFont{mathrsfs}{OT1}{rsfs}{m}{n}
\DeclareSymbolFontAlphabet{\mathrsfs}{mathrsfs}
\newcommand*{\supp}{{\rm supp}\,}

\newcommand*{\ket}[1]{\left| #1 \right\rangle}
\newcommand*{\bra}[1]{\left\langle{#1}\right|}

\newtheorem{theorem}{Theorem}
\newtheorem{corollary}{Corollary}

\begin{document}
\title{Shapes of leading tunnelling trajectories for single-electron molecular ionization} 

\author{Denys I. Bondar}
\email{dbondar@princeton.edu}
\affiliation{Department of Physics and Astronomy and Guelph-Waterloo Physics Institute, University of Waterloo, Waterloo, Ontario N2L 3G1, Canada}
\affiliation{National Research Council of Canada, Ottawa, Ontario K1A 0R6, Canada}
\affiliation{Present Address:  Frick Laboratory, Department of Chemistry, Princeton University, Princeton, New Jersey 08544, USA}

\author{Wing-Ki Liu}
\email{wkliu@uwaterloo.ca}
\affiliation{Department of Physics and Astronomy and Guelph-Waterloo Physics Institute, University of Waterloo, Waterloo, Ontario N2L 3G1, Canada}
\affiliation{Department of Physics, The Chinese University of Hong Kong,
Shatin, NT, Hong Kong}

\begin{abstract}
Based on the geometrical approach to tunnelling by P.D. Hislop and I.M. Sigal [Memoir. AMS {\bf 78},  No. 399 (1989)], we introduce the concept of a leading tunnelling trajectory. It is then proven that leading tunnelling trajectories for single-active-electron models of molecular tunnelling ionization (i.e., theories where a molecular potential is modelled by a single-electron multi-centre potential) are linear in the case of short range interactions and ``almost'' linear in the case of long range interactions. The results are presented on both the formal and physically intuitive levels. Physical implications of the obtained results are discussed.  
\end{abstract}


\maketitle

\section{Introduction}

Recent advances in experimental investigations of single-electron molecular ionization in a low frequency strong laser field \cite{Talebpour1996, Guo1998, DeWitt2001, Wells2002, Litvinyuk2003, Pavivcic2007, Kumarappan2008, Staudte2009, Hoff2010} have created a demand for a theory of this phenomenon. There is a broad variety of theoretical approaches to this problem: molecular extensions of the analytical atomic strong-field methods \cite{MuthBohm2000, Tong2002, Kjeldsen2004, Kjeldsen2006, Fabrikant2009, Fabrikant2010, Murray2010, Bin2010, Walters2010, Murray2011, Murray2011a}, numerical methods explicitly incorporating tunnelling (the Floquet approach \cite{Madsen1998, Chu2000}, the complex scaling method \cite{Plummer1996}, and the method of complex absorbing potentials \cite{Otobe2004}), numerical solutions of the time-dependent Schr\"{o}dinger equation within the single-active-electron approximation \cite{Chelkowski1992, Awasthi2008, Luhr2008, Petretti2010, Abu-samha2009, Abu-samha2010, Spanner2010}, numerical solutions of the time-dependent Schr\"{o}dinger equation for two-electron systems \cite{Saenz2000, Harumiya2002, Awasthi2006, Vanne2008, Vanne2009, Vanne2010}, and treatments based on the time-dependent density functional theory \cite{Chu2004, Telnov2009, Fowe2010, Chu2011}.

As far as low frequency laser radiation is concerned, one can ignore the time-dependence of the laser and consider the corresponding quasistatic picture, which is obtained in the limit the laser frequency $\omega\to 0$. In this limit, single electron molecular ionization is realized by quantum tunnelling. This approximation is valid from qualitative and quantitative points of view, and it tremendously simplifies the theoretical analysis of the problem at hand. Such single-active-electron approaches to molecular ionization, where an electron is assumed to interact with multiple centres that model the molecule and a static field that models the laser, are among the most popular. Analytical and semi-analytical versions of these methods, which are based on the quasiclassical approximation \cite{Tong2002, Fabrikant2009, Murray2010, Bin2010, Fabrikant2010, Murray2011} are indeed quite successful in interpreting and explaining available experimental data. However, these quasiclassical theories heavily rely on the assumption that the electron tunnels along a straight trajectory. The purpose of the current paper is to study the reliability of this hypothesis. 

Relying on the geometrical approach to many dimensional tunnelling by Hislop and Sigal \cite{Sigal1988a, Sigal1988, Hislop1989a, Hislop1996}, which is a mathematically rigorous reformulation of the instanton method, we first introduce the notion of leading tunnelling trajectories. Then, we analyze their shapes in the context of single-active-electron molecular tunnelling. It will be rigorously proven that the assumption of ``almost'' linearity of leading tunnelling trajectories is satisfied in almost all the situations of practical interest. Such results justify the above mentioned models, and perhaps, open new ways of further development of quasiclassical approaches to molecular ionization.

The rest of the paper is organized as follows: Section \ref{Sec2} is a concise  introduction to the Hislop and Sigal geometrical ideas and related topics. The proof of the results regarding the shapes of the leading tunnelling trajectories are presented in Sec. \ref{Sec3}. We employ multiple spherically symmetric potential wells, which is the simplest type of molecular potentials, to estimate single electron molecular tunnelling rates in Sec. \ref{Sec4}. Leading tunnelling trajectories are numerically computed in Sec. \ref{Sec5} for model diatomic molecules of different geometries, and the rule of thumb on how to find the shapes of leading tunnelling trajectories is formulated. Conclusions are drawn in the last section. Finally, the Appendix contains the derivation of the multidimensional generalization of the Landau method of calculating quasiclassical matrix elements.

\section{Mathematical background}\label{Sec2}

The instanton approach is one of the methods for the description of tunnelling \cite{Holstein1996}. It can be introduced as a result of application of the saddle point approximation to the modification of the Feynman integral obtained by performing the transformation of time $t\to -i\tau$ to ``imaginary time'' $\tau$ (i.e., the Wick rotation). This technique has turned out to be tremendously fruitful in many branches of physics and chemistry (see, e.g., Refs. \cite{Vainshtein1982, Leggett1984, Benderskii1994, Razavy2003a}). 

We shall reiterate the main steps in deriving the instanton approach. Let us consider a quantum system with the Hamiltonian  
\begin{eqnarray}\label{InstSec_Hamiltonian}
\hat{H} = -\Delta/(2m) + U({\bf x}),
\end{eqnarray}
where $\Delta$ is the $n$-dimensional Laplacian and ${\bf x}$ is an $n$-dimensional vector. The Feynman integral representation of the propagator reads \cite{Feynman1965} (atomic units are used throughout, unless stated otherwise) 
\begin{eqnarray}
&& \bra{ {\bf x}_f } e^{-i\hat{H}t_0} \ket{{\bf x}_i} = N \int \mathrsfs{D}[{\bf x}(t)] e^{ i S[{\bf x}(t)] }, \label{FeynmanRepPropag}\\
&& S[{\bf x}] = \int_0^{t_0} \mathrsfs{L}({\bf x}, \dot{\bf x}) dt, 
\quad \mathrsfs{L}({\bf x}, \dot{\bf x}) = \frac{\dot{\bf x}^2}{2m} - U({\bf x}), \nonumber
\end{eqnarray}
where the path integral sums up all the paths that obey boundary conditions ${\bf x}(0) = {\bf x}_i$ and ${\bf x}(t_0) = {\bf x}_f$, and $\dot{\bf x}(t) \equiv d {\bf x}(t) /dt$. After performing the Wick rotation, Eq. (\ref{FeynmanRepPropag}) becomes 
\begin{eqnarray}\label{EuclidianFeynmanRepProp}
&& \bra{ {\bf x}_f } e^{-\hat{H}\tau_0} \ket{{\bf x}_i} = N \int \mathrsfs{D}[{\bf x}(\tau)] e^{ -\tilde{S}[{\bf x}(\tau)] }, \\
&& \tilde{S}[{\bf x}] = \int_0^{\tau_0} \left[ \frac{1}{2m} \left(\frac{d{\bf x}(\tau)}{d\tau}\right)^2 + U({\bf x}(\tau))  \right]d\tau, \nonumber
\end{eqnarray}
where $\tau_0 = it_0$ and $\tilde{S}$ is called {\it the Euclidian action}.  
Hence, one can say that the transition from Eq. (\ref{FeynmanRepPropag}) to Eq. (\ref{EuclidianFeynmanRepProp}) is achieved by the following formal substitutions
\begin{eqnarray}
 t \to -i\tau, \quad {\bf x}(t) \to {\bf x}(\tau), \quad \dot{\bf x}(t) \to id{\bf x}(\tau)/d\tau. \label{WickRotationETC}
\end{eqnarray}
Comparing the actions $S$ and $\tilde{S}$, one concludes that the motion in imaginary time is equivalent to the motion in the inverted potential. In other words, the actions $S$ and $\tilde{S}$ are connected by the substitution
\begin{eqnarray}
U \to -U, \qquad (E \to -E). \label{FlippingPotentialEnergy}
\end{eqnarray}
The final step in the instanton approach is the application of the saddle point approximation to the Euclidian Feynman integral in Eq. (\ref{EuclidianFeynmanRepProp}) assuming that $\tau_0 \to \infty$.

However, there is a long ongoing discussion \cite{Dunne2000, Eldad1977, Mueller-Kirsten2001, VanBaal1986} whether the instanton approach agrees with the quasiclassical approximation for tunnelling; some observations have been made that these two methods may disagree up to a pre-exponential factor. Furthermore,  as it has been pointed out in Ref. \cite{Leggett1984}, the instanton approach in the formulation presented so far [substitutions (\ref{WickRotationETC})] not only looks like a ``highly dubious manoeuvre,'' but also gives no prescription for getting a correct pre-exponential factor. In authors' opinion, such discrepancies mainly occur because the Feynman integral is just a heuristic construction without a sound mathematical ground \cite{Schulman2005}. The absence of strict rules of calculation of the Feynman integral makes impossible any definitive judgment of a particular result. Consequently,  a natural question arises how the instanton method can be safely used and what the meaning of the substitutions (\ref{WickRotationETC}) and (\ref{FlippingPotentialEnergy}) is.

The mathematical physics community has reinterpreted the instanton approach rigorously (see, e.g., Refs. \cite{Agmon1982, Carmona1981, Helffer1988, Herbst1993, Hislop1989a, Hislop1996, Sigal1988a, Sigal1988} and references therein), and the corresponding analysis answers both questions. Moreover, this rigorous interpretation is extremely useful because it can be implemented as an effective numerical method, which will lead to a clear physical picture applicable to a broad class of problems. We shall review briefly the above cited works since on the one hand, they are often unfamiliar to physicists, and on the other hand, they may be challenging to read for non-specialists in mathematical physics.

Historically, the first problem considered within such a framework was ``how fast does a bound state decay at infinity?'' \cite{Agmon1982, Carmona1981} (see also Sec. 3 of Ref. \cite{Hislop1996}). Let us clearly pose the question. Consider the Hamiltonian (\ref{InstSec_Hamiltonian}) as a self-adjoint operator on $\mathrsfs{L}_2(\mathbb{R}^n)$ -- the space of square-integrable functions. A bound state wave function $\psi \in \mathrsfs{L}_2(\mathbb{R}^n)$ is a normalizable eigenfunction of such a Hamiltonian, $\hat{H}\psi = E\psi$. Since the normalization integral converges, the bound state wave function $\psi = \psi({\bf x})$ must vanish as  $\| {\bf x} \| \to \infty$. Therefore, we want to determine how this decay is affected by the potential $U$. This question can be answered very elegantly if we confine ourself to {\it an upper bound} on the rate of decay.

To obtain this upper bound, we need to introduce first some geometrical notions. Let $M$ be a real $n$-dimensional manifold (intuitively, $M$ is some $n$-dimensional surface). The tangent space at a point ${\bf x}\in M$, denoted by $T_{\bf x}(M)$, is a real linear vector space $\mathbb{R}^n$ that intuitively contains all the possible ``directions'' in which one can tangentially pass through ${\bf x}$. A metric is an assignment of an inner (scalar) product to $T_{\bf x}(M)$ for every ${\bf x} \in M$.

Let ${\bf x} \in \mathbb{R}^n$ and ${\bm \xi}, {\bm \eta} \in T_{\bf x}(M)$. We define a (degenerate) metric by
\begin{eqnarray}\label{AgmonMetricDeff}
\langle {\bm \xi}, {\bm \eta} \rangle_{\bm x} \equiv 2m(U({\bf x}) - E)_+ \langle {\bm \xi}, {\bm \eta} \rangle,
\end{eqnarray}
where $\langle {\bm \xi}, {\bm \eta} \rangle \equiv {\bm \xi}\cdot {\bm \eta} = \xi_1 \eta_1 + \ldots + \xi_n \eta_n$ is the Euclidean inner product and $f({\bf x})_+ \equiv \max\{ f({\bf x}), 0\}$. Following the convention used in mathematical literature, we shall call metric (\ref{AgmonMetricDeff}) as {\it the Agmon metric}.

Having introduced the metric, we can equip the manifold $M$  with many geometrical notions such as distance, angle, volume, etc. The length of a differentiable path ${\bm \gamma} : [0,1] \to \mathbb{R}^n$ in the Agmon metric is defined by
\begin{eqnarray}\label{AgmonLengthOfCurve}
L_A({\bm \gamma}) &=& \int_0^1 \| \dot{\bm \gamma}(t) \|_{{\bm \gamma}(t)} dt \nonumber\\
&=& \sqrt{2m} \int_0^1 [U({\bm \gamma}(t)) - E]_+^{1/2} \| \dot{\bm \gamma}(t) \| dt,
\end{eqnarray}
where $\| {\bm \xi} \| = \sqrt{\langle {\bm \xi}, {\bm \xi} \rangle}$ is the Euclidian (norm) length, and $\| {\bm \xi} \|_{\bf x} = \sqrt{\langle {\bm \xi}, {\bm \xi} \rangle_{\bf x}}$. The path of a minimal length is called a geodesic. Finally,  {\it the Agmon distance between points ${\bf x}, {\bf y} \in \mathbb{R}^n$}, denoted by $\rho_E({\bf x}, {\bf y})$, is the length of the shortest geodesic in the Agmon metric connecting ${\bf x}$ to ${\bf y}$. 

Before going further, we would like to clarify the physical meaning of the Agmon metric. Let us recall {\it the Jacobi theorem} from classical mechanics (see page 150 of Ref. \cite{Gustafson2003} and page 247 of Ref. \cite{Arnold1989}): The classical trajectories of the system with the potential $U({\bf x})$ and  a total energy $E$ are geodesics in {\it the Jacobi metric}
\begin{eqnarray}\label{JacobiMetricDeff}
\langle\langle {\bm \xi}, {\bm \eta} \rangle\rangle_{\bf x} = 2m(E - U({\bf x}))_+ \langle {\bm \xi}, {\bm \eta} \rangle,
\end{eqnarray} 
on the set $\{ {\bf x} \in \mathbb{R}^n | U({\bf x}) \leqslant E\}$ --  the classical allowed region. The Agmon metric [Eq. (\ref{AgmonMetricDeff})] and the Jacobi metric [Eq. (\ref{JacobiMetricDeff})] are indeed connected through the substitution (\ref{FlippingPotentialEnergy}). By virtue of this analogy, we conclude that the Agmon distance has to satisfy a time-independent Hamilton-Jacobi equation, also known as an eikonal equation,
\begin{eqnarray}\label{Hamiljacobi_AgmonDistance}
| \nabla_{\bf x} \rho_E({\bf x}, {\bf y}) |^2 = 2m( U({\bf x}) - E )_+,
\end{eqnarray}
where $\nabla_{\bf x} f({\bf x}) \equiv (\partial f /\partial x_1, \ldots, \partial f/\partial x_n)$. In fact, the Agmon distance is the Euclidean version of the reduced action [here, the adjective ``Euclidian'' means the same as in Eq. (\ref{EuclidianFeynmanRepProp})]. In other words, the Agmon distance is the action of an instanton. 

Now we are in position to recall upper bounds on a bound eigenstate of the Hamiltonian (\ref{InstSec_Hamiltonian}). First, 
under very mild assumptions on $U$ (continuity, compactness of the classically allowed region, and absence of tunnelling, i.e., the spectrum of the Hamiltonian being only real), it has been proven \cite{Agmon1982} that for an arbitrary small $\epsilon > 0$, there exists a constant $0<c_{\epsilon}<\infty$, such that
\begin{eqnarray}\label{MildUpperBoundOnBoundState}
\int e^{2(1-\epsilon)\rho_E({\bf x})} |\psi({\bf x})|^2 d^n {\bf x} \leqslant c_{\epsilon},
\end{eqnarray}
where $\rho_E({\bf x}) \equiv \rho_E({\bf x}, {\bf 0})$. Roughly speaking, result (\ref{MildUpperBoundOnBoundState}) means that $\psi({\bf x}) = O\left( e^{-(1-\epsilon)\rho_E({\bf x})} \right)$. However, this result can be improved.  For any small $\epsilon > 0$,  there exists a constant $0<c_{\epsilon}<\infty$, such that the following inequality is valid under additional conditions of regularity of the potential $U$
\begin{eqnarray}\label{UpperBoundOnBoundState}
| \psi({\bf x}) | \leqslant c_{\epsilon} e^{-(1-\epsilon)\rho_E({\bf x})}.
\end{eqnarray}

Analyzing Eq. (\ref{MildUpperBoundOnBoundState}) and Eq. (\ref{UpperBoundOnBoundState}), we conclude that the Agmon distance from the origin describes the exponential factor of the wave function. Further information can be found in Refs. \cite{Agmon1982,  Herbst1993, Hislop1996} and references therein. We note that lower bounds on ground states can also be obtained by utilizing the Agmon approach \cite{Carmona1981}.  

We illustrate the power and utility of the upper bound (\ref{UpperBoundOnBoundState}) by deriving upper bounds for matrix elements and transition amplitudes in the Appendix. The former result is an estimate of the modulo square of the matrix element 
$ 
\bra{\psi_p} V \ket{\psi_q}, 
$ 
where $\psi_p$ and $\psi_q$ are bound eigenstates of the Hamiltonian (\ref{InstSec_Hamiltonian}) that correspond to eigenvalues $E_p$ and $E_q$. It is demonstrated in the Appendix that for an arbitrary small $\epsilon > 0$, there exists a constant $0<c_{\epsilon}<\infty$, such that
\begin{eqnarray}\label{Inequality_LandauQuasiclassicalMatrixElem}
\left| \bra{\psi_p} V \ket{\psi_q} \right|^2 \leqslant c_{\epsilon}\int V^2 ({\bf x}) e^{-2(1-\epsilon)\left[ \rho_{E_p}({\bf x}) + \rho_{E_q}({\bf x}) \right]} d^n {\bf x},
\end{eqnarray}
which could be interpreted as, 
\begin{eqnarray}\label{LandauQuasiclassicalMatrixElemGenral}
&& \left| \bra{\psi_p} V \ket{\psi_q} \right|^2 \nonumber\\
&& \quad  = O\left( \int V^2({\bf x}) e^{-2(1-\epsilon)\left[ \rho_{E_p}({\bf x}) + \rho_{E_q}({\bf x}) \right]} d^n {\bf x} \right).
\end{eqnarray}

Simplicity of the derivation of Eq. (\ref{LandauQuasiclassicalMatrixElemGenral}) does not imply its insignificance. On the contrary, Eq. (\ref{LandauQuasiclassicalMatrixElemGenral}) is a multidimensional generalization of the Landau method of calculating quasiclassical matrix elements \cite{Landau1932} (see also page 185 of Ref. \cite{Landau_1977} and Refs. \cite{Nikitin1991, Nikitin1993}). To the best of authors' knowledge, such a generalization has not been reported before. To prove the one-dimensional version of the Landau method using analytical techniques (as it is usually done), one deals with the Stokes phenomenon (see, e.g., Ref. \cite{Meyer1989}); thus, the generalization to the multidimensional case without too restrictive assumptions is not obvious. The Agmon upper bounds lead not only to quite a trivial derivation, but also to an intuitive physical and geometrical picture. 

Now we explain briefly how these geometrical ideas are generalized to the problem of tunnelling (interested readers should consult Refs. \cite{Hislop1989a, Hislop1996, Sigal1988, Sigal1988a} and references therein for details and further development). Let $E$ be an energy of a tunnelling particle. We denote the boundary of the classically  forbidden region by $S_E$. It is assumed that $S_E$ consists of two disjoint pieces  $S_E^-$ and $S_E^+$ (i.e., $S_E = S_E^- \cup S_E^+$ and $S_E^- \cap S_E^+ = \emptyset$) -- the inside and outside turning surfaces, which are merely multidimensional analogs of turning points. Having introduced the concept of the Agmon distance, we naturally introduce two related notions: First, {\it the Agmon distance from the surface $S_E^-$ to a point ${\bf x}$, $\rho_E({\bf x}, S_E^-)$},  as the minimal Agmon distance between the point ${\bf x}$ and an arbitrary point ${\bf y} \in S_E^-$ [more rigorously, $
\rho_E({\bf x}, S_E^-) =  \inf_{{\bf y} \in S_E^-}  \rho_E({\bf x}, {\bf y})$]; second, {\it the Agmon distance between the turning surfaces $S_E^-$ and $S_E^+$, $\rho_E(S_E^-, S_E^+)$,} as the minimal Agmon distance between arbitrary two points ${\bf x} \in S_E^+$ and ${\bf y} \in S_E^-$ [ $\rho_E(S_E^-, S_E^+) = \inf_{{\bf x} \in S_E^+} \rho_E({\bf x}, S_E^-)$]. 

In a nutshell, and thus a bit abusing the formulation of the original result \cite{Hislop1989a}, we say that for an arbitrary small $\epsilon>0$, there exists a constant $c>0$, such that the tunnelling rate, $\Gamma$, (viz., the width of  a resonance) in the quasiclassical limit ($\hbar\to 0$) obeys 
\begin{eqnarray}\label{HislopSigal_UpperBound}
\Gamma \leqslant c\exp[ - 2\beta_E (\tilde{\rho}_E - \epsilon) ],
\end{eqnarray}
where $0 < \tilde{\rho}_E < \infty$ and $\beta_E\tilde{\rho}_E$ being the leading asymptote of $\rho_E(S_E^-, S_E^+)$ when $\hbar\to 0$. However, the following interpretation of upper bound (\ref{HislopSigal_UpperBound}) is sufficient for our further applications:
\begin{eqnarray}\label{AgmonTunnellingRate}
\Gamma = O\left( e^{-2\rho_E(S_E^-, S_E^+) } \right),
\end{eqnarray}
i.e., twice the Agmon distance between the turning surfaces gives the leading exponential factor of the tunnelling rate within the quasiclassical approximation.

The Agmon distance between two points, $\rho_E({\bf x}, {\bf y})$, can be computed by solving numerically Eq. (\ref{Hamiljacobi_AgmonDistance}) with the boundary condition 
\begin{eqnarray}
\rho_E( {\bf y}, {\bf y}) = 0
\end{eqnarray}
by means of the fast marching method \cite{Kimmel1998, Barth1998a, Sethian1999, Kimmel2004, Hassouna2007}. Moreover, having computed the solution, one can readily extract the minimal geodesic from a given initial point  ${\bf x}$ by back propagating along $\rho_E({\bf x}, {\bf y})$, where ${\bf y}$ is regarded as a fixed parameter; more explicitly, the minimal geodesic, ${\bf g} \equiv {\bf g}(t)$, is obtained as the solution of the following Cauchy problem \cite{Kimmel1998, Kimmel2004}
\begin{eqnarray}\label{MinGeodesic}
\dot{{\bf g}} = -\nabla_{\bm \xi} \rho_E({\bm \xi}, {\bf y}), \qquad {\bf g}(0) = {\bf x}.
\end{eqnarray}
Such a geodesic can be interpreted as a ``tunnelling trajectory.''

A brief remark on types of the solutions of Eq. (\ref{Hamiljacobi_AgmonDistance}) ought to be made. Generally speaking, an eikonal equation admits a local solution under reasonable assumptions, but a global solution is not possible in a general case owing to the possibility of development of caustics (see, e.g., Ref. \cite{Evans1998}). Nonetheless, when we talk about a solution of Eq. (\ref{Hamiljacobi_AgmonDistance}), we actually refer to a viscosity solution because not only it is a global solution, but also it has the meaning of distance \cite{Sethian1999, Kimmel2004} which we originally assigned to the function $\rho_E$. Another reason for employing only the viscosity solution of the eikonal equation is as follows: Writing the wave function as $\Psi({\bf x}) = \exp[ - S(\hbar; {\bf x}) ]$, the time-independent Schr\"{o}dinger equation becomes 
$$
	\left| \nabla_{\bf x} S(\hbar; {\bf x}) \right|^2 - \hbar \Delta S(\hbar; {\bf x}) = 2m( U({\bf x}) - E ).
$$
Comparing this equation with Eq. (\ref{Hamiljacobi_AgmonDistance}) in the classical forbidden region, we conclude that
$$
	\rho_E = \lim_{\hbar\to 0} S(\hbar; {\bf x}),
$$
which is the definition of the viscosity solution of eikonal equation (\ref{Hamiljacobi_AgmonDistance}) (see, e.g., page 540 of Ref. \cite{Evans1998}).

In fact, the fast marching method is an ``upwind'' finite difference method that efficiently computes the viscosity solution of an eikonal equation. Note, hence, that the fast marching method as well as the other ideas presented and developed in the current paper cannot be employed to study the influence of chaotic tunnelling trajectories (see Ref. \cite{Levkov2009} and references therein). Some implementations of the fast marching method as well as the minimal geodesic tracing can be downloaded from Refs. \cite{KroonAccurateFM, PeyreToolboxFM, ChuLSMLIB}.
  
The Agmon distance from the surface to a point, $\rho_E({\bf x}, S_E^-)$, must satisfy Eq. (\ref{Hamiljacobi_AgmonDistance}). Indeed, $\rho_E({\bf x}, S_E^-)$ is the solution of the boundary problem
\begin{eqnarray}\label{AgmonDistanceFromSurf_HJ}
&& |\nabla_{\bf x} \rho_E({\bf x}, S_E^-) |^2 = 2m(U({\bf x}) - E)_+, \\
&& \rho_{E} ({\bf y}, S_E^-) = 0, \qquad \forall {\bf y} \in S_E^-, \nonumber
\end{eqnarray}
which can be solved by the fast marching method as well. Finally, the Agmon distance between the turning surfaces is computed as $\min_{{\bf x} \in S_E^+} \rho_E({\bf x}, S_E^-)$ after solving Eq. (\ref{AgmonDistanceFromSurf_HJ}). 

The points ${\bf b} \in S_E^-$ and ${\bf e} \in S_E^+$ such that
\begin{eqnarray}
\rho_E (S_E^-, S_E^+) = \rho_E ({\bf b}, {\bf e}),
\end{eqnarray}
 are of physical importance because they represent the points where the particle ``begins'' its motion under the barrier (${\bf b}$) and ``emerges'' from the barrier (${\bf e}$), correspondingly. Moreover, the minimal geodesic (\ref{MinGeodesic}) that connects these points (${\bf g}(0) = {\bf b}$ and ${\bf g}(1) = {\bf e}$) is  a tunnelling trajectory which gives the largest tunnelling rates -- the {\it leading tunnelling trajectory}. Note, however, that these points as well as the trajectories may not be unique in a general case.
 
The idea of utilization of upper bounds to describe tunnelling is not new. Kapur and Peierls \cite{Kapur1937, Kapur1937a} (see also Ref. \cite{Peierls1991}) have proposed as early as 1937 that even though many dimensional quasiclassical approximation is untractable in its original formulation, it still can be used to obtain the upper bound on the probability of transmission through a barrier. The geometrical ideas reviewed in the current section can be viewed upon as a reincarnation of the Kapur-Peierls approach with an important (and convenient for our applications) emphasis on the geometrical aspect of the method. 

It is also noteworthy that a power of the fast marching method in applications to tunnelling has already been recognized in chemistry within the context of the reaction path theory \cite{Dey2004, Dey2006a, Dey2006b, Dey2007, Liu2010}. Similarly to the current paper, the main object of interest of those studies is the reaction path, which is the leading tunnelling trajectory in our terminology. Nevertheless, the motivation for the usage of the fast marching method, presented in Refs.  \cite{Dey2004, Dey2006a, Dey2006b, Dey2007, Liu2010}, is tremendously  different from the geometrical point of view adopted here. 

\section{Main Results}\label{Sec3}

In this section, we shall follow a two step program. First, we consider tunnelling in multiple finite range potentials, where we prove that leading tunnelling trajectories are linear (Theorem \ref{theorem1}). Then, we reduce the case of multiple long range potentials to the previous one by employing the fact that a singular long range potential can be represented as a sum of a singular short range potential and a continuous long range tail [Eq. (\ref{PartitionLongRangeSingularPotential})]. Such a reduction allows us to prove that the leading tunnelling trajectories are ``almost'' linear (Theorem \ref{theorem2}). We note that partitioning (\ref{PartitionLongRangeSingularPotential}) was put forth by Perelomov, Popov, and Terent'ev \cite{Perelomov1966a, Perelomov67a, Perelomov1967a, Popov1968}, and it is widely used for obtaining the Coulomb corrections in strong filed ionization  (see Refs. \cite{Popov2004, Popov2005, Smirnova2008a, Bondar2009, Popruzhenko2008a, Popruzhenko2008b, Popruzhenko2009} and references therein).

Let us introduce some notations. Hereinafter, the dimension of the space is assumed to be $n\geqslant 2$.  The interaction of an electron with a static electric field of the strength $F$ is of the form $Fx_n$ ($F>0$). $\partial A$ denotes the boundary of the region $A$. The map, ${\bm \min_{x_n}} : \mathbb{R}^n \supset A \to \mathbb{R}^n$, selects a point ${\bf x} = {\bm \min_{x_n}} A \in A$ that has the smallest $x_n$ component among all the other points from $A$, assuming that $A$ has such a unique point. The projection $P{\bf x}$ of the point ${\bf x} = (x_1, x_2, \ldots, x_n)$ is defined as $P{\bf x} = (x_1, \ldots, x_{n-1}, E/F)$.

\begin{theorem}\label{theorem1}
We study single electron tunnelling ($-\infty< E<0$, $F>0$) in the potential 
\begin{eqnarray}\label{TotalPotentialComplectSupport}
U({\bf x}) = \sum_{j=1}^K V_j(\| {\bf x} - {\bf R}_j \|) + Fx_n.
\end{eqnarray}
Let us assume that 
\begin{enumerate}
\item $V_j : (0, R_j) \to (-\infty, 0)$ and $V_j : (R_j, +\infty) \to \{ 0 \}$, $R_j > 0$, $j=1,\ldots,K$, are differentiable on $(0, R_j)$ and strictly increasing functions, such that $V_j(0) = -\infty$ and $V_j$ may have a jump discontinuity at the point $R_j$. 
\item $\supp V_j  = \left\{ {\bf x} \in  \mathbb{R}^n \, | \, V_j (\| {\bf x} - {\bf R}_j \|) \neq 0 \right\}$ is the support of the potential $V_j (\| {\bf x} - {\bf R}_j \|)$, $\supp V_k \cap \supp V_j = \emptyset$, $\forall k \neq j$ and $\supp V_j \cap \left\{  {\bf x} \in  \mathbb{R}^n \, | \,  x_n \leqslant E/F \right\} = \emptyset$, $j = 1,\ldots,K$.
\item Introduce ${\bf q}_j = {\bm \min_{x_n}} \partial \supp V_j$, ${\bf p}_j =  {\bm \min_{x_n}} S_E^-(j)$, $S_E^-(j)$ is defined in Eq. (\ref{SEminusJ_Deff}). If there exists $N$, such that 
\begin{eqnarray}\label{EuclidianAssumption}
\| {\bf p}_N - P{\bf R}_N \| < \| {\bf q}_j - P{\bf R}_j \|, \quad \forall j \neq N,
\end{eqnarray}
\end{enumerate}
Then, the leading tunnelling trajectory is unique and linear, and it starts at the point ${\bf p}_N$ and ends at $P{\bf R}_N$, $\rho_E( S_E^-, S_E^+) = \rho_E({\bf p}_N,  P{\bf R}_N)$.
\end{theorem}
\begin{proof}
 The boundary of the classically forbidden region is defined by the equation $U({\bf x}) = E$. Consider two cases: 
 
 First, if ${\bf x} \notin \bigcup_{j=1}^K \supp V_j $ then according to assumption 2, the equation $U({\bf x}) = E$ simply reads $Fx_n = E$, and thus its solution defines the outer turning surface $S_E^+ = \left\{  {\bf x} \in  \mathbb{R}^n \, | \,  x_n = E/F \right\}$. One can see now that the projector operator $P$ projects a point onto $S_E^+$. 
 
 Second, if ${\bf x} \in \supp V_j$ and $V_j$ is continuous at the point $R_j$, then the equation reads $V_j (\| {\bf x} - {\bf R}_j \|) + Fx_n = E$. To proof that the set 
 \begin{eqnarray}\label{SEminusJ_Deff}
 S_E^-(j) = \left\{ {\bf x} \in \supp V_j \, | \,  V_j (\| {\bf x} - {\bf R}_j \|) + Fx_n = E \right\}
 \end{eqnarray}
 is not empty, we construct the function $f_j( {\bf x} ) = V_j (\| {\bf x} - {\bf R}_j \|) + Fx_n -E$. Since $f_j({\bf R}_j) = -\infty$, we can find a set $A_j \subset \supp V_j$ located close to ${\bf R}_j$, such that $f_j({\bf x}) < 0$ for all ${\bf x} \in A_j$; correspondingly, since according to assumption 2, $x_n > E/F$, there exists the set $B_j \subset \supp V_j$ of points close to the boundary of $\supp V_j$ for which $f_j$ is positive. In fact, $A_j$ and $B_j$ can be constructed such that $\| {\bf x} - {\bf R}_j \| < \| {\bf y} - {\bf R}_j \|$, $\forall {\bf x} \in A_j$ and $\forall {\bf y} \in B_j$. Therefore, the intermediate value theorem guarantees that  $S_E^-(j) \neq \emptyset$ and it ``lies between'' $A_j$ and $B_j$. Furthermore, the inner turning surface is $S_E^- = \bigcup_{j=1}^K S_E^-(j)$, and $S_E^-(j) \cap S_E^-(k) = \emptyset$, $\forall j\neq k$. (Note that the strict monotonicity of $V_j(x)$ assures that the set $S_E^-(j)$ is connected.) Whence, 
\begin{eqnarray}\label{ReductionManyCenterCaseToOneCenter}
\rho_E( S_E^-, S_E^+) = \min_j \left\{ \rho_E (S_E^-(j), S_E^+) \right\}.
\end{eqnarray}
Eq. (\ref{ReductionManyCenterCaseToOneCenter}) means the reduction of the many centre case  to the singe centre case under the assumptions made. Needles to mention that such a reduction tremendously simplifies the analysis. 

The same conclusions are valid if the jump of the function $V_j$ at $R_j$ is not too large, so that the equation $V_j (\| {\bf x} - {\bf R}_j \|) + Fx_n = E$ has solutions for $ {\bf x} \in \supp V_j$. However, if the jump is too large, i.e., this equation does not have solutions from the support of the potential, then it is natural to set $S_E^-(j) = \partial \supp V_j$.

Consider the single centre case -- single electron tunnelling in the potential $U_j ({\bf x}) = V_j (\| {\bf x} - {\bf R}_j \|) + Fx_n$. We shall show that this potential is axially symmetric. If ${\bf x}= (x_1, \ldots, x_n)$, then we introduce $\Pi{\bf x} \equiv (x_1, \ldots, x_{n-1})$. We can then symbolically write ${\bf x} = (\Pi{\bf x} , x_n)$. Using this new notation, we obtain 
\begin{eqnarray}\label{Rewritten_Potential_Uj}
U_j ({\bf x}) = V_j \left( \sqrt{\| \Pi{\bf x} - \Pi{\bf R}_j\|^2 + \left( x_n - \left[{\bf R}_j\right]_n \right)^2} \right) + F x_n,
\end{eqnarray}
where $\left[ {\bf a} \right]_n$ denotes the $n^{\rm th}$ component of the vector ${\bf a}$. It is readily seen from Eq. (\ref{Rewritten_Potential_Uj}) that the potential $U_j ({\bf x})$ is invariant under transformations that do not change $x_n$ and arbitrary $(n-1)$ dimensional (proper and improper) rotations of the vector $\Pi{\bf x}$ about the point $\Pi{\bf R}_j$. The only invariant subspace of $\mathbb{R}^n$ under such transformations is the line $\{ (\Pi{\bf R}_j, x_n) \, | \, x_n \in \mathbb{R}\}$.

Since both regions $S_E^-(j)$ and $S_E^+$ are shape invariant under the axial symmetry transformations, we may expect that the shortest geodesic connecting these regions ought to be shape invariant as well. 
Thus, one readily concludes that the leading tunnelling trajectory should be linear and should connect the points ${\bf p}_j$ and  $P{\bf R}_j$
\begin{eqnarray}\label{S_minus_j_equality}
\rho_E ({\bf p}_j, P{\bf R}_j ) = \rho_E (S_E^-(j), S_E^+),
\end{eqnarray}
since no other geodesic that connects $S_E^-(j)$ and $S_E^+$ is shape invariant with respect to the axial symmetry transformations. Below we shall present a formal version of this derivation.

Foremost, we demonstrate that the operation ${\bm \min_{x_n}}$ is defined on the set $S_j^-(j)$, viz., that there is a unique point of $S_j^-(j)$ that has the smallest component $x_n$. Employing the method of Lagrange multipliers and taking into account the symmetry of the potential, we construct the function 
\begin{eqnarray}
&& \mathrsfs{L}_1 (x_n, c, \lambda) = x_n \nonumber\\
&& \quad + \lambda \left[ V_j\left( \sqrt{ c^2 + \left(x_n - \left[{\bf R}_j\right]_n\right)^2 }\right)  + Fx_n -E\right]. \quad
\end{eqnarray}
The condition $\partial \mathrsfs{L}_1 /\partial c = 0$ leads to $c=0$. Therefore, 
${\bf p}_j =  {\bm \min_{x_n}} S_E^-(j) = (\Pi{\bf R}_j, y)$, where $y$ being the minimal solution of the equation 
\begin{eqnarray}\label{Eq.ForY}
V_j\left(\left| y - \left[{\bf R}_j\right]_n\right|\right) + Fy = E. 
\end{eqnarray}
Moreover, $P{\bf p}_j \equiv P{\bf R}_j \equiv P{\bf q}_j$.

Eq. (\ref{Eq.ForY}) must have two distinct solutions $y_{1,2}$ ($y_1 < y_2$). $y_1$ ($y_2$) corresponds to the point from $S_E^-(j)$ with the minimum (maximum) $x_n$. Additionally, since $E -F y_1 > E-Fy_2$ $\Rightarrow$ $V_j\left(\left| y_1 - \left[{\bf R}_j\right]_n\right|\right) > V_j\left(\left| y_2 - \left[{\bf R}_j\right]_n\right|\right)$, we obtain 
\begin{eqnarray}\label{InequalitiesForY}
\eta_j \equiv \left| y_1 - \left[{\bf R}_j\right]_n\right| > \left| y_2 - \left[{\bf R}_j\right]_n\right|.
\end{eqnarray}

To find the maximum of the function $\| {\bf x} - {\bf R}_j \|$ on the set $S_E^-(j)$ within the Lagrange multipliers method, we introduce the function
\begin{eqnarray}\label{LagrangeMultipliers2}  
&& \mathrsfs{L}_2 (x_n, c, \lambda) = \sqrt{c^2 + \left(x_n - \left[{\bf R}_j\right]_n\right)^2} \nonumber\\
&& \quad + \lambda \left[ V_j\left( \sqrt{ c^2 + \left(x_n - \left[{\bf R}_j\right]_n\right)^2 }\right)  + Fx_n -E\right]. \quad
\end{eqnarray}
Taking into account inequality (\ref{InequalitiesForY}) and the fact that $\partial \mathrsfs{L}_2 /\partial c = 0$ $\Rightarrow$ $c=0$, we conclude that the maximum of the function $\| {\bf x} - {\bf R}_j \|$ on $S_E^-(j)$ is reached at the point ${\bf p}_j$.

Let $S_j( z )$ denote a sphere of the radius $z$ centred at ${\bf R}_j$, $S_j(z) = \{ {\bf x} \in \mathbb{R}^n \, | \, \| {\bf x} - {\bf R}_j \| = z \}$. Consider a sequence of spheres
$
\left\{ S_j\left(\eta_j + k[R_j-\eta_j]/W \right) \right\}_{k=0}^W, 
$
where $S_j (R_j) = \partial \supp V_j$ and $\eta_j$ was introduced in Eq. (\ref{InequalitiesForY}). Now pick a sequence of points, $\{ {\bm \gamma}( k/W ) \}_{k=0}^W$, such that, ${\bm \gamma}( k/W ) \in  S_j\left(\eta_j + k[R_j-\eta_j]/W \right)$, $k=0,
\ldots,W$. We assume that this sequence is a discretization of some  differentiable path ${\bm \gamma} : [0,1] \to \mathbb{R}^n$. According to Eq. (\ref{AgmonLengthOfCurve}), the sums,
\begin{eqnarray}\label{SigmaWDeff}
\Sigma_W({\bm \gamma}) &=& \sqrt{2m} \sum_{k=0}^W \sqrt{ U_j( {\bm \gamma}(k/W)) - E}\nonumber\\
&& \times \| {\bm \gamma}( [k+1]/W ) - {\bm \gamma}(k/W) \|,
\end{eqnarray}
where we set ${\bm \gamma}( 1+1/W ) \equiv {\bm \gamma}( 1 )$, obeys the property $\lim_{W\to\infty} \Sigma_W ({\bm \gamma}) = L_A({\bm \gamma})$. Introduce a path:
\begin{eqnarray}\label{PathGDeff}
{\bf g}(t) = {\bf p}_j+ t\left[ {\bf q}_j - {\bf p}_j \right].
\end{eqnarray}
Since $\forall k$, ${\bf g}(k/W) \in  S_j\left(\eta_j + k[R_j-\eta_j]/W \right)$, $ \left[{\bm \gamma}(k/W)\right]_n  \geqslant \left[{\bf g}(k/W)\right]_n$ and
$V_j (\| {\bf g}(k/W) - {\bf R}_j \| ) = V_j (\| {\bm \gamma}(k/W) - {\bf R}_j \|) $ $\Rightarrow$ $U_j ({\bm \gamma}(k/W)) \geqslant U_j({\bf g}(k/W))$. Moreover, $ \| {\bm \gamma}( [k+1]/W ) - {\bm \gamma}(k/W) \| \geqslant  \| {\bf g}( [k+1]/W ) - {\bf g}(k/W) \|$. Therefore, 
\begin{eqnarray}
\Sigma_W({\bm \gamma}) \geqslant \Sigma_W({\bf g}) \Rightarrow 
L_A ({\bm \gamma}) \geqslant L_A({\bf g}).
\end{eqnarray}
Since $\Sigma_W({\bm \gamma}) = \Sigma_W({\bf g}) \Leftrightarrow {\bm \gamma}(k/W) = {\bf g}(k/W)$, $k = 0,\ldots,W-1$, $\forall W$, we conclude that path (\ref{PathGDeff}) is indeed the shortest geodesic that connects the regions $S_E^-(j)$ and $\partial\supp V_j$. By the same token, the geodesic connecting $\partial\supp V_j$ and $S_E^+$ must be a straight line that starts at ${\bf q}_j$ and ends at $P{\bf q}_j$ because the potential between these two regions is merely $V({\bf x}) =  Fx_n$.

To finalize the proof, we shall {\it ``backward propagate''} the leading tunnelling trajectory starting from the outer turning surface $S_E^+$. Let $\tilde{\rho}({\bf x}, {\bf y})$ denote the Agmon distance between two points for the potential $V({\bf x}) =  Fx_n$. Then, it is easy to demonstrate that 
\begin{eqnarray}\label{AgmonDIstanceInConstantField}
\tilde{\rho}_E({\bf x}, P{\bf x}) = (2/3)\sqrt{2mF} \| {\bf x} - P{\bf x} \|^{3/2}.
\end{eqnarray}
The plane $T(c) = \{ {\bf x} \in \mathbb{R}^n \, | \, x_n = c \}$ is a surface of constant Agmon distance [Eq. (\ref{AgmonDIstanceInConstantField})], such that $\tilde{\rho}_E(T_{E/F}, S_E^+) = 0$ and $\tilde{\rho}_E(T(c), S_E^+)$ is a strictly increasing function of $c$. Since $\| {\bf q}_N - P{\bf R}_N \| = \| {\bf p}_N - P{\bf R}_N \| - \| {\bf p}_N - {\bf q}_N \| < \| {\bf q}_j - P{\bf R}_j \|$,  $\forall j\neq N$, condition (\ref{EuclidianAssumption}) guarantees that increasing $c$ the plane $T(c)$ will ``hit'' the boundary of $\supp V_N$ at the point ${\bf q}_N$. (Note that $\tilde{\rho}_E\left(T(c), S_E^+\right) \equiv \rho_E\left(T(c), S_E^+\right)$, $E/F < c  < \left[ {\bf q}_N \right]_n$.) Moreover,  the following follows from Eq. (\ref{EuclidianAssumption})
$$
\{ {\bf x} \in \mathbb{R}^n \, | \,\left[ {\bf q}_N \right]_n \leqslant x_n \leqslant \left[ {\bf p}_N \right]_n \}\cap\supp V_j = \emptyset, \quad \forall j \neq N,
$$
which means that $N^{\rm th}$ centre is isolated from all the other. Therefore, the shortest geodesic must connect the point ${\bf q}_N$ to the point ${\bf p}_N$.
\end{proof}

\begin{corollary}\label{corollary1}
Consider a single electron tunnelling in the potential (\ref{TotalPotentialComplectSupport}), such that assumptions 1 and 2 of Theorem \ref{theorem1} are satisfied, then the leading trajectory is linear (but may not be unique).
\end{corollary}
\begin{proof} 
This corollary follows from the straightforward generalization of the idea of backward propagation. 
\end{proof}

\begin{theorem}\label{theorem2}
We shall study single electron tunnelling ($-\infty< E<0$, $F>0$) in the potential 
\begin{eqnarray}\label{TotalPotentialLongRangeCase}
U({\bf x}) = \sum_{j=1}^K \mathrsfs{V}_j(\| {\bf x} - {\bf R}_j \|) + Fx_n.
\end{eqnarray}
Assume that
\begin{enumerate}
\item $\mathrsfs{V}_j : (0, +\infty) \to (-\infty, 0)$ are differentiable on $(0, +\infty)$ and strictly increasing functions, such that $\mathrsfs{V}_j(0)=-\infty$ and $\mathrsfs{V}_j(+\infty)=0$.
\item The boundary of the classically forbidden region consists of two disjoint pieces -- the internal turning surface $S_E^-$ and  the outer one $S_E^+$. Furthermore, $S_E^- = \bigcup_{j=1}^K S_E^-(j)$, $S_E^-(j) \cap S_E^-(k) = \emptyset$, $\forall j\neq k$, where each $S_E^-(j)$ encircles ${\bf R}_j$\footnote{
The verb ``encircle'' should be understood in the following sense: A piece of the inner turning surface, $S_E^-(j)=\partial CA(j)$, is a boundary of the classically allowed region, $CA(j)$, associated with centre $j$, such that ${\bf R}_j \in CA(j)$.
}. 
\item $B(j) \cap B(k) = \emptyset$, $\forall j\neq k$, and $B(j) \cap S_E^+ = \emptyset$, $\forall j$, where $B(j) = \left\{ {\bf x} \in  \mathbb{R}^n \, |\,  \| {\bf x} - {\bf R}_j \| \leqslant r_j \right\}$ being the ball of radius $r_j$ centered at ${\bf R}_j$. Here $r_j = \max\left\{ \| {\bf x} - {\bf R}_j \| \, | \, {\bf x} \in S_E^-(j) \right\}$ is the ``radius'' of $S_E^-(j)$\footnote{
The parameter $r_j$ can be calculated by means of the method of Lagrange multipliers as it was shown in the proof of Theorem \ref{theorem1} [see Eq. (\ref{LagrangeMultipliers2})].
}.
\end{enumerate}
Then, the leading tunnelling trajectory (may not be unique) is linear up to a term of $O(\lambda)$ as $\lambda\to 0$, where $\lambda = \max_j \left\{ |\mathrsfs{V}_j(\Delta_j) |\right\}$ and 
\begin{eqnarray}\label{DeltaDeff}
\Delta_j = \min\left( \frac{r_j}2 + \frac 12 \min_{k, \, k \neq j} \left\{ \| {\bf R}_j - {\bf R}_k \| - r_k \right\}, d_j \right).
\end{eqnarray}
Here, $d_j = \min\left\{ \| {\bf x} - {\bf R}_j \| \, | \, {\bf x} \in S_E^+ \right\}$ is the Euclidean distance from ${\bf R}_j$ to $S_E^+$.
\end{theorem}

\begin{proof}
We introduce two auxiliary functions
\begin{eqnarray}
V_{sh}^{(j)}(x) &=& \left\{ 
	\begin{array}{lll}
	\mathrsfs{V}_j(x) & : & 0 \leqslant x < \Delta_j, \\
	0 & : & x \geqslant \Delta_j,
	\end{array} 
\right. \nonumber\\
V_{lg}^{(j)}(x) &=& \left\{ 
	\begin{array}{lll}
	0 & : & 0 \leqslant x < \Delta_j, \\
	\mathrsfs{V}_j(x)  & : & x \geqslant \Delta_j.
	\end{array} 
\right. \nonumber
\end{eqnarray}
One evidently notices that
\begin{eqnarray}\label{PartitionLongRangeSingularPotential}
\mathrsfs{V}_j(x) = V_{lg}^{(j)}(x) + V_{sh}^{(j)}(x),
\end{eqnarray}
where $V_{sh}^{(j)}(x)$ is a singular short range potential and $V_{lg}^{(j)}(x)$ being a long range tail. The purpose of such a partition is to make $V_{sh}^{(j)}(x)$ satisfy assumption 1 of Theorem \ref{theorem1} and produce $V_{lg}^{(j)}(x)$ that obeys the following upper bound:
$$ 
|V_{lg}^{(j)}(x)| \leqslant  |\mathrsfs{V}_j(\Delta_j) | \leqslant \lambda, \qquad \forall x.
$$ 

We analyze the length of a curve in the Agmon metric [Eq. (\ref{AgmonLengthOfCurve})]. Since 
\begin{eqnarray}
\sqrt{ U({\bf x}) - E } = \sqrt{ \sum_{j=1}^K V_{sh}^{(j)}( \| {\bf x} - {\bf R}_j \| ) + Fx_n -E + O(\lambda)} \nonumber\\
= \sqrt{ \sum_{j=1}^K V_{sh}^{(j)}( \| {\bf x} - {\bf R}_j \| ) + Fx_n -E} + O(\lambda), \nonumber
\end{eqnarray}
under the assumption that $\lambda\to 0$, we have reduced the initial situation to the case of single electron tunnelling in the potential 
\begin{eqnarray}
U_{sh}({\bf x}) = \sum_{j=1}^K V_{sh}^{(j)}(\| {\bf x} - {\bf R}_j \|) + Fx_n.
\end{eqnarray}

Let us now utilize assumption 3 to show that
\begin{eqnarray}\label{DeltaJBiggerRJ}
\Delta_j > r_j.
\end{eqnarray}
Indeed, on the one hand, $B(j)\cap S_E^+ = \emptyset$ $\Rightarrow$ $d_j > r_j$; on the other hand, $B(j) \cap B(k) = \emptyset$, $\forall j\neq k$, $\Rightarrow$ $\| {\bf R}_j - {\bf R}_k \| - r_k > r_j$. 

Furthermore, we shall demonstrate that the definition of  $\Delta_j$ [Eq. (\ref{DeltaDeff})] assures that assumption 2 of  Theorem \ref{theorem1} for the functions $V_{sh}^{(j)}(x)$ holds. According to Eq. (\ref{DeltaDeff}), 
$$
\Delta_j  \leqslant \left(  \| {\bf R}_j - {\bf R}_k \| - r_k + r_j \right)/2, \qquad 
j \neq k;
$$
hence, $\Delta_j + \Delta_k \leqslant \| {\bf R}_j - {\bf R}_k \|$ $\Rightarrow$ $\supp V_{sh}^{(j)} \cap \supp V_{sh}^{(k)} = \emptyset$. From Eq. (\ref{DeltaDeff}), we also obtain $\Delta_j \leqslant d_j$ $\Rightarrow$ $\supp V_{sh}^{(j)} \cap S_E^+ = \emptyset$; thus, the outer turning surface for the potential $U_{sh}({\bf x})$ should be $\left\{  {\bf x} \in  \mathbb{R}^n \, | \,  x_n = E/F \right\}$.

Finally, we have proven the theorem because the potential $U_{sh}({\bf x})$ satisfies all the assumptions of Corollary \ref{corollary1}. 
\end{proof}

Physical clarifications of Theorems \ref{theorem1} and \ref{theorem2} are due. Assumption 1 of Theorem \ref{theorem1} physically implies that $V_j$ are attractive, singular, spherically symmetric short range potentials.  Assumption 2 of the same theorem requires that the potentials do not merge, i.e., their ranges do not overlap. This condition connotes that the classically allowed regions associated with the centres ${\bf R}_j$ [their boundaries are $S_E^-(j)$] do not overlap as well. The latter statement is proven in Theorem \ref{theorem1}. The statement of Corollary \ref{corollary1} can be rephrased as follows: leading tunnelling trajectories for a system of non-overlapping, attractive, singular, short range potentials are linear. However, if the additional condition (\ref{EuclidianAssumption}) is satisfied then Theorem \ref{theorem1} not only guarantees the uniqueness of the leading tunnelling trajectory, but also provides the initial and final points of the trajectory. Assumption 1 of Theorem \ref{theorem2}  means that $\mathrsfs{V}_j$ are attractive, singular, spherically symmetric long range  potentials that vanish at infinity.  Assumptions 2 and 3 of Theorem \ref{theorem2} require the same non-overlapping condition for the classically allowed internal regions mentioned above. Physically, Theorem \ref{theorem2} says that leading tunnelling trajectories for a system of several such potentials are ``almost'' linear, and a deviation from being strictly linear is caused by vanishing long tails of the potentials; thus, the larger the distance between the centres, the smaller the deviation. 

In a nutshell, all these results have been achieved because the multi centre (i.e., molecular) potential is represented as a sum of spherically symmetric potentials, and such conclusions regarding the shape of the trajectories in the single centre (i.e., atomic) case are quite expectable owing to the axial symmetry. 

An important case not covered by the theorems is the case of overlapping potentials that physically corresponds to valence electrons, which form chemical bonds and have a low ionization potential and are delocalized over a molecule. This case as well as the issue of uniqueness of the trajectories will be scrutinized in Sec. \ref{Sec5}.

\section{The Application of Spherically Symmetric Potential Wells to Single Electron Molecular Tunnelling}\label{Sec4}

The simplest type of model molecular potentials that allows for full analytical treatment is of type (\ref{TotalPotentialComplectSupport}) where
\begin{eqnarray}\label{PotentialWell}
V_j(x) = \left\{
	\begin{array}{ccc}
		-\infty & : & 0 < x < r_j, \\
		0 & : & x > r_j.
	\end{array}
\right.
\end{eqnarray}
It is assumed that $S_E^-(j) = \partial \supp V_j = \{ {\bf x} \in  \mathbb{R}^n \, | \, \| {\bf x} \| = r_j \}$.  (Strictly speaking, these potentials are not governed by Theorem \ref{theorem1}.) Evidently, the leading tunnelling trajectories are linear, and moreover, the following equality is valid 
\begin{eqnarray}\label{AgmonDistanceForWells}
\rho_E ( S_E^-, S_E^+) = \min_j \left\{ \tilde{\rho}_E \left({\bf q}_j, P{\bf q}_j \right) \right\},
\end{eqnarray}
where ${\bf q}_j = {\bm \min}_{x_n} S_E^-(j)$ and $\tilde{\rho}_E$ was defined in Eq. (\ref{AgmonDIstanceInConstantField}). Let us estimate the tunnelling rates within Eq. (\ref{AgmonTunnellingRate}) for the two dimensional system of two equivalent centres of type (\ref{PotentialWell}) (see Fig. \ref{FigTwoCentres}).  A straightforward geometrical derivation, using Eqs. (\ref{AgmonTunnellingRate}), (\ref{AgmonDIstanceInConstantField}), and (\ref{AgmonDistanceForWells}), shows that
\begin{eqnarray}\label{TwoEqualPotentialWellsTunnelRates}
\Gamma \propto \exp\left\{ -\frac 2{3F} \left[ FR(1-| \cos\theta |) - 2E \right] ^{3/2} \right\},
\end{eqnarray} 
where $R$ is the distance between the potential wells (i.e., the bond length of a model molecule) and $\theta$ is the angle between the field and the molecular axis. The obtained angular dependent rates are plotted in Fig. \ref{FigRatesTwoCentres}. 

According to Eq. (\ref{AgmonTunnellingRate}), rates obtained within the geometrical approach do not account for an initial molecular orbital. This technique provides solely the contribution  from the shape of the barrier, hence, the name -- the ``geometrical approach.'' An advantage of such a method is that it reduces the calculation of tunnelling rates to a rather simple geometrical exercise.
 
\begin{figure}
\begin{center}
\includegraphics[scale=0.30]{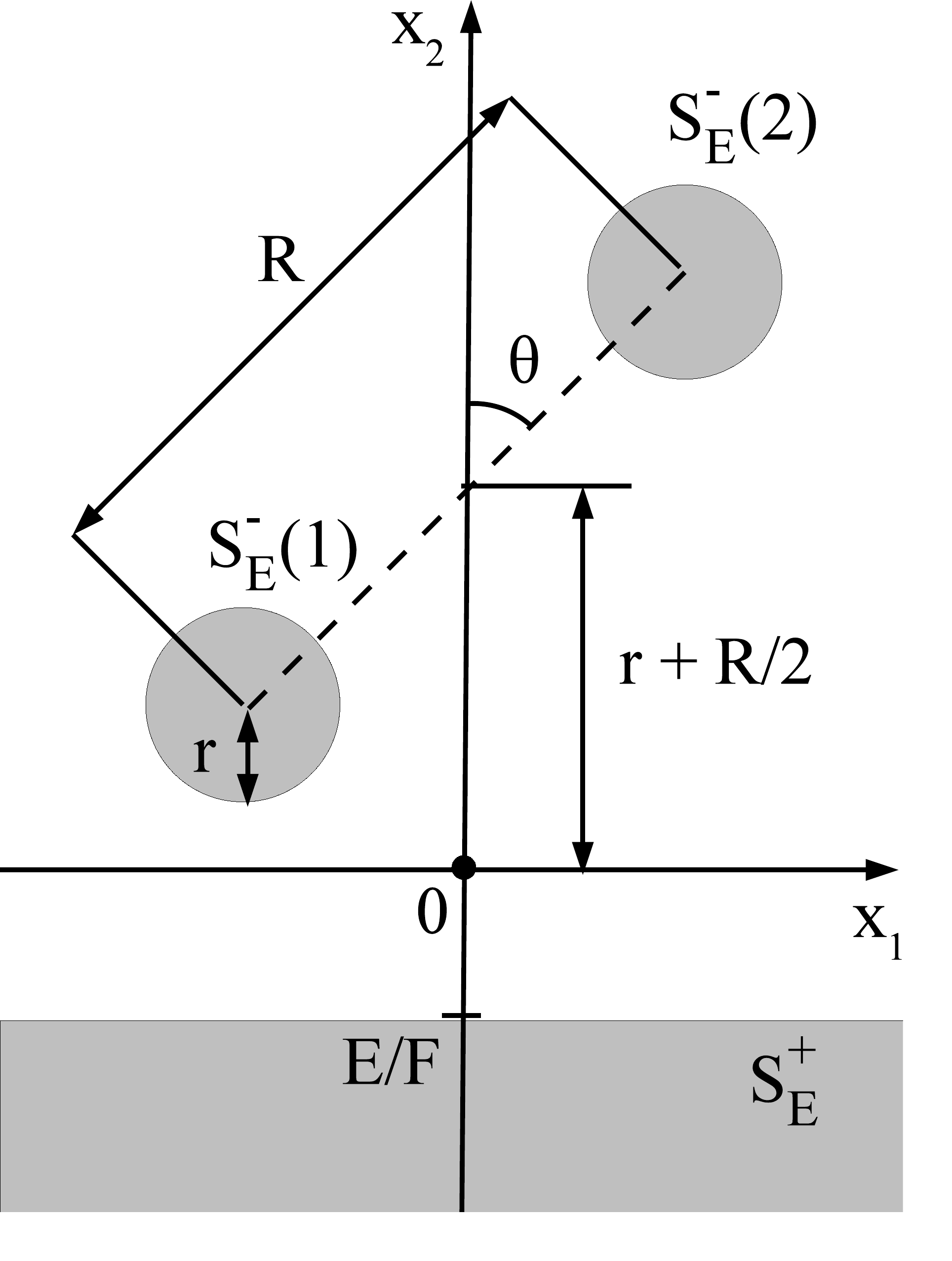}
\caption{The geometry of a two centre model employed to obtain Eq. (\ref{TwoEqualPotentialWellsTunnelRates}). Grey colour denotes the classically allowed regions.}\label{FigTwoCentres}
\end{center}
\end{figure}

\begin{figure}
\begin{center}
\includegraphics[scale=0.50]{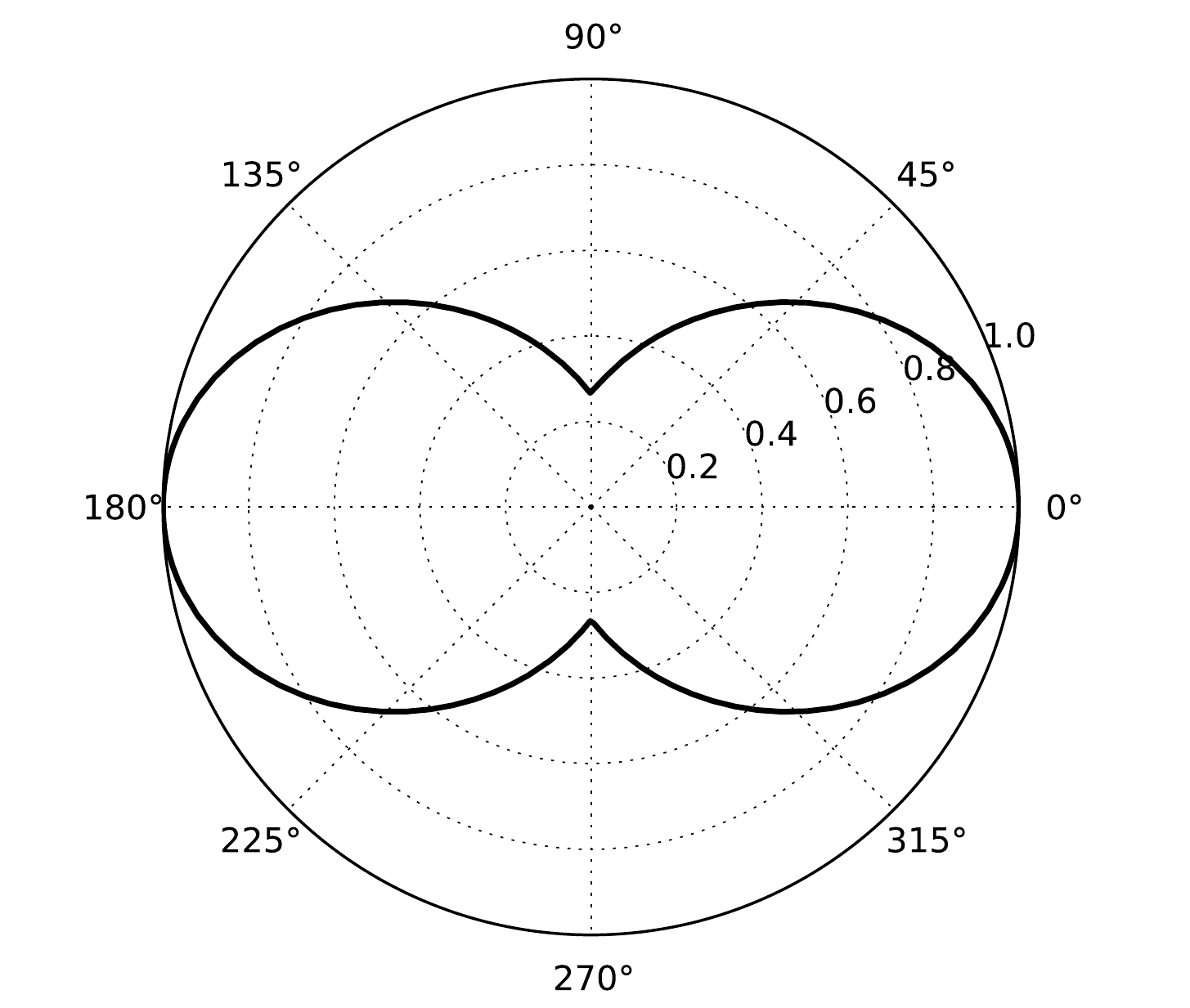}
\caption{The polar plot of the normalized tunnelling rates for the two centre model [Eq. (\ref{TwoEqualPotentialWellsTunnelRates})] as a function of the angle $\theta$. Chosen values of the parameters are $F=0.01$ (a.u.), $R=2$ (a.u.), and $E=-0.5$ (a.u.).}\label{FigRatesTwoCentres}
\end{center}
\end{figure}

\section{Numerical Illustrations}\label{Sec5}

In order to illustrate the results of Theorem \ref{theorem2} and also draw some conclusions beyond Theorem \ref{theorem2}, we shall calculate the shapes of leading tunnelling trajectories for different situations within the numerical scheme sketched in Sec. \ref{Sec2}. To achieve a good accuracy, we compute the viscosity solutions of eikonal equation (\ref{Hamiljacobi_AgmonDistance}) by means of the second order multi-stencil fast marching method \cite{Hassouna2007}. The following two-dimensional model potential is used for a diatomic molecule in this section:
\begin{eqnarray}
	U (x_1, x_2) &=& -\sum_{j=1}^2 Z_j \left[ \left(x_1 - \xi^{(j)}_1\right)^2 + \left(x_2 - \xi^{(j)}_2 \right)^2 \right]^{-1/2} \nonumber\\
					&& + F x_2, \label{LongRangeDiatomic} 
\end{eqnarray}
where the first atom is centred at $(\xi_1^{(1)}, \xi_2^{(1)})$ and the second atom is at $(\xi_1^{(2)}, \xi_2^{(2)})$,
\begin{eqnarray}
	\xi_1^{(1)} = - (R/2) \sin\theta, \qquad \xi_2^{(1)} = 4 - (R/2) \cos\theta, \nonumber\\
	\xi_1^{(2)} =  (R/2) \sin\theta, \qquad \xi_2^{(2)} = 4 + (R/2) \cos\theta, \nonumber\\	
	F=0.05 \mbox{ (a.u.)}, \quad E = -0.5 \mbox{ (a.u.)}, \quad m=1 \mbox{ (a.u.)}. 
\end{eqnarray}
Regarding the definitions of the interatomic distance $R$ and the angle $\theta$ between the molecular axis and the external field $F$, see Fig. \ref{FigTwoCentres}.

Solutions of the eikonal equation and leading tunnelling trajectories for identical non-overlapping long range potentials (\ref{LongRangeDiatomic}) are pictured in Fig. \ref{FigNonOverlapingLongRange}. Figure \ref{FigNonOverlapingLongRange}(a) illustrates the non-uniqueness of leading tunnelling trajectories in the case of non-overlapping potentials. The reason of this non-uniqueness is the mirror symmetry of the potential with respect to the axis $x_1 = 0$, as a result, the probabilities of tunnelling along each trajectory coincide. Note that the leading tunnelling trajectories in Fig. \ref{FigNonOverlapingLongRange}(a) are almost linear. Results shown in Figs. \ref{FigNonOverlapingLongRange}(b) and  \ref{FigNonOverlapingLongRange}(c) are in full agreement with the statement of Theorem \ref{theorem2}.

Leading tunnelling trajectories for identical overlapping long range potentials are shown in Fig. \ref{FigOverlapingLongRange}. The shapes of all these trajectories are almost linear as well. Non-uniqueness of leading tunnelling trajectories in the case of overlapping potentials is demonstrated in Fig. \ref{FigOverlapingLongRange}(a). However, if the interatomic distance is further decreased, one observes in Fig. \ref{FigOverlapingLongRange}(d) that the previous two distinct trajectories merge into one 
restoring the uniqueness of the leading tunnelling trajectory.

The shapes of leading tunnelling trajectories in the case of non-identical long range potentials are shown in Fig. \ref{NonIdenticalLongRange}. The interatomic distance is chosen such that one observes a transition between non-overlapping and overlapping cases by simply rotating the molecule. We clearly see in Fig. \ref{NonIdenticalLongRange} that leading tunnelling trajectories are almost linear. 

The results presented in all these figures can be summed up in the following {\it rule of thumb on how to find the leading tunnelling trajectory:} The final point (or multiple points when more than one leading tunnelling trajectory is possible) is near the point on the outside turning curve (surface in the $n \geqslant 3$ dimensional case), $S_E^+$, with the smallest value of $x_2$ ($x_n$ in the $n$ dimensional case). The initial point is usually near the point on the inside turning curve, $S_E^-$, with the smallest value of $x_2$. The tangent vector to the leading tunnelling trajectory at the initial and final points tend to be perpendicular to the inside and outside turning curves, respectively. 

Note that the rule of thumb is valid for polyatomic molecules and for an arbitrary number of dimensions. 

\begin{figure}
\begin{center}
	\includegraphics[scale=0.45]{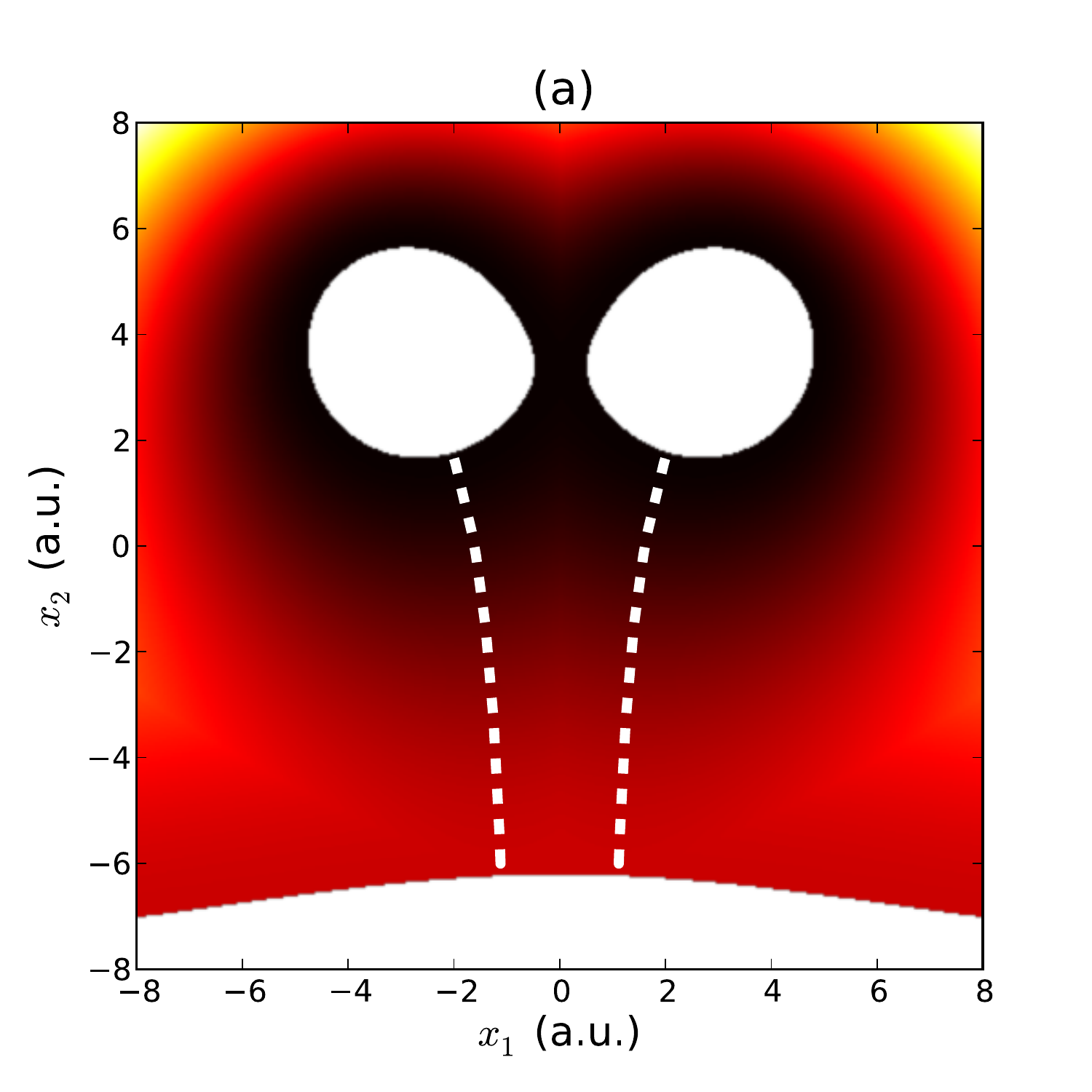}
	\includegraphics[scale=0.45]{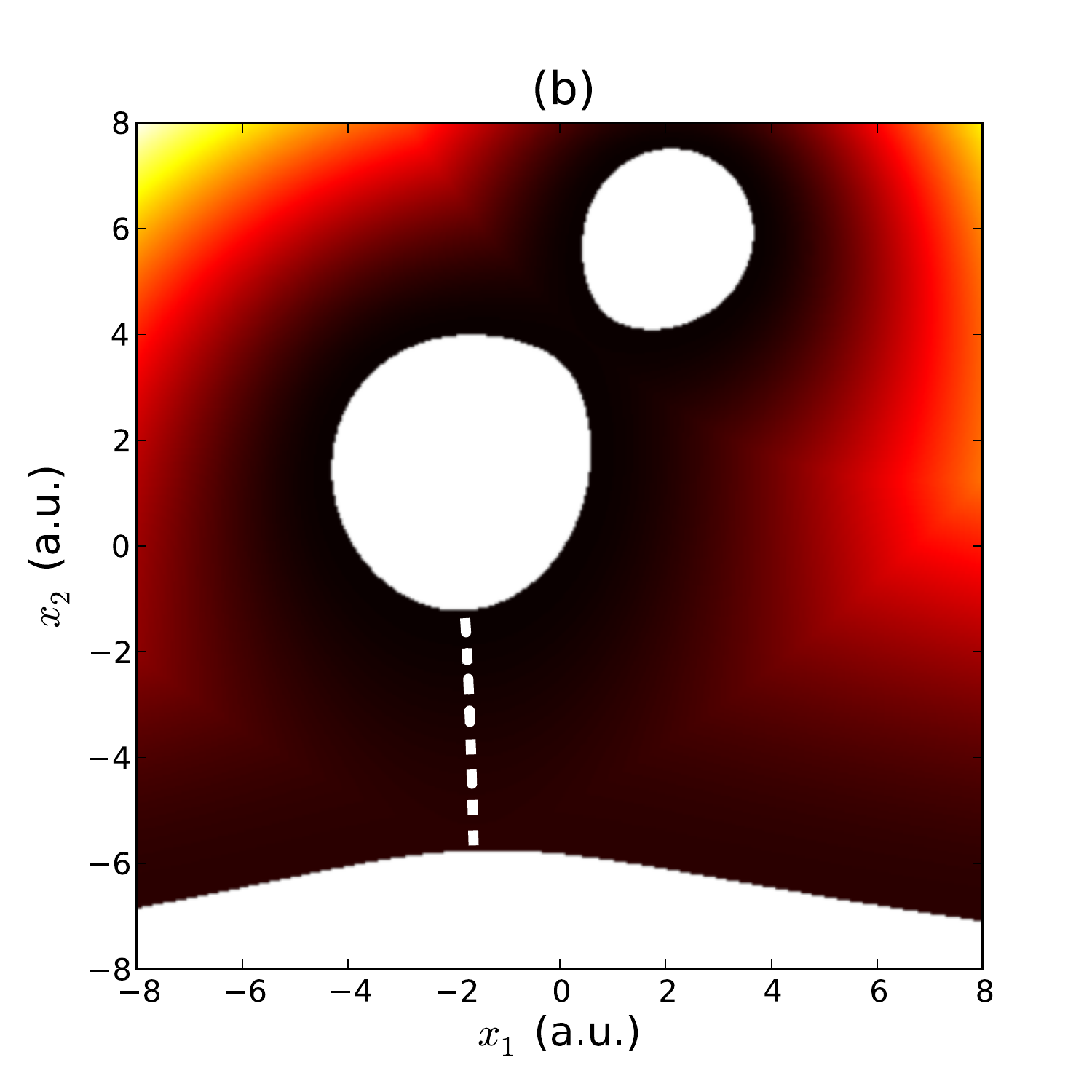}
	\includegraphics[scale=0.45]{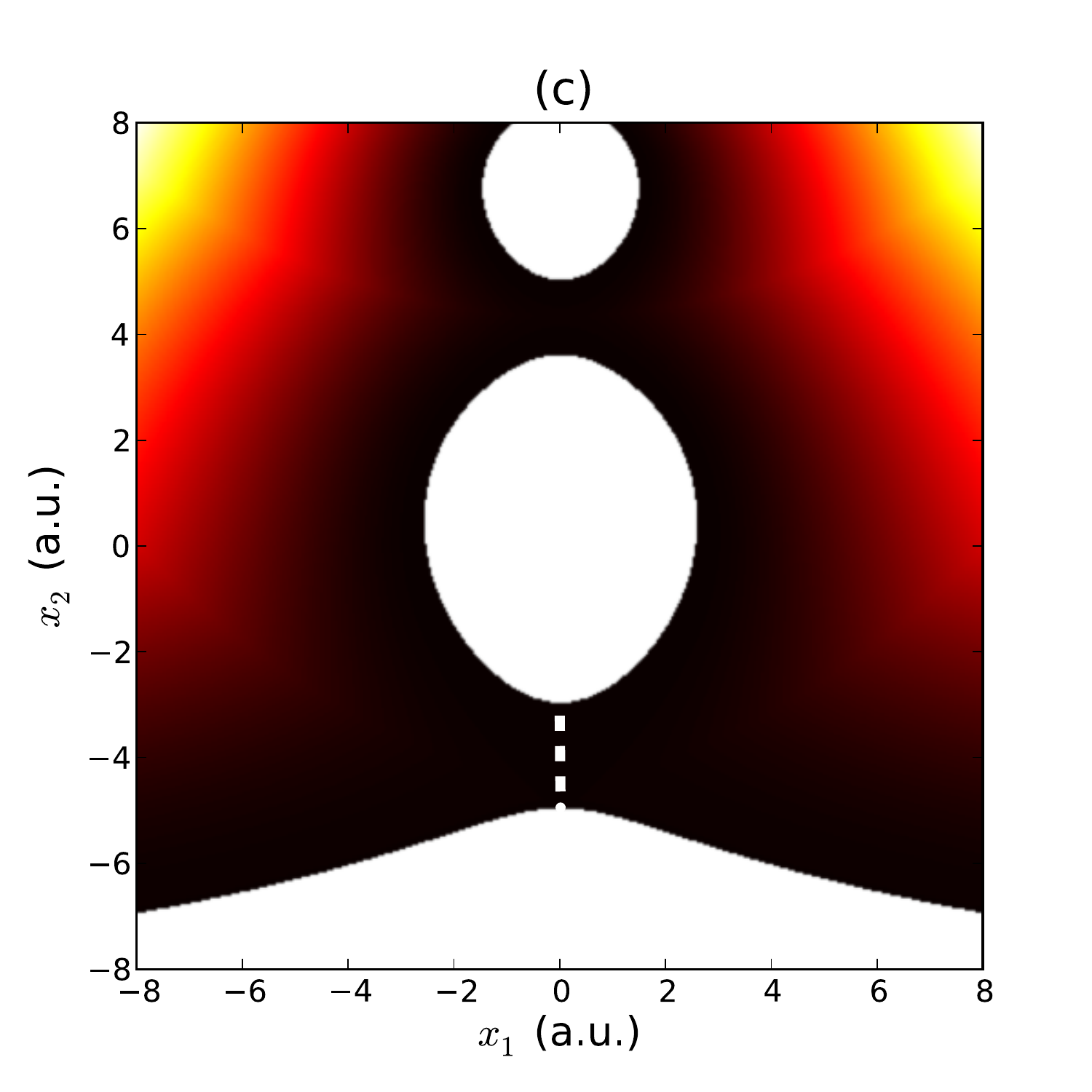}
\caption{(Color online) Viscosity solutions of eikonal equation (\ref{Hamiljacobi_AgmonDistance}) and leading tunnelling trajectories in the case of two identical non-overlapping long range potentials, which are given by Eq. (\ref{LongRangeDiatomic}) with $Z_1 = Z_2 = 1$ and $R = 6$ a.u. White colour denotes the classically allowed regions. Dashed white lines are leading tunnelling trajectories. The solutions of the eikonal equation is represented by linear scale colour ramps from black (minimum) to bright color (maximum). (a) $\theta = 90^{\circ}$; (b) $\theta = 45^{\circ}$; (c) $\theta = 0^{\circ}$.}\label{FigNonOverlapingLongRange}
\end{center}
\end{figure}

\begin{figure*}
\begin{center}
	\includegraphics[scale=0.45]{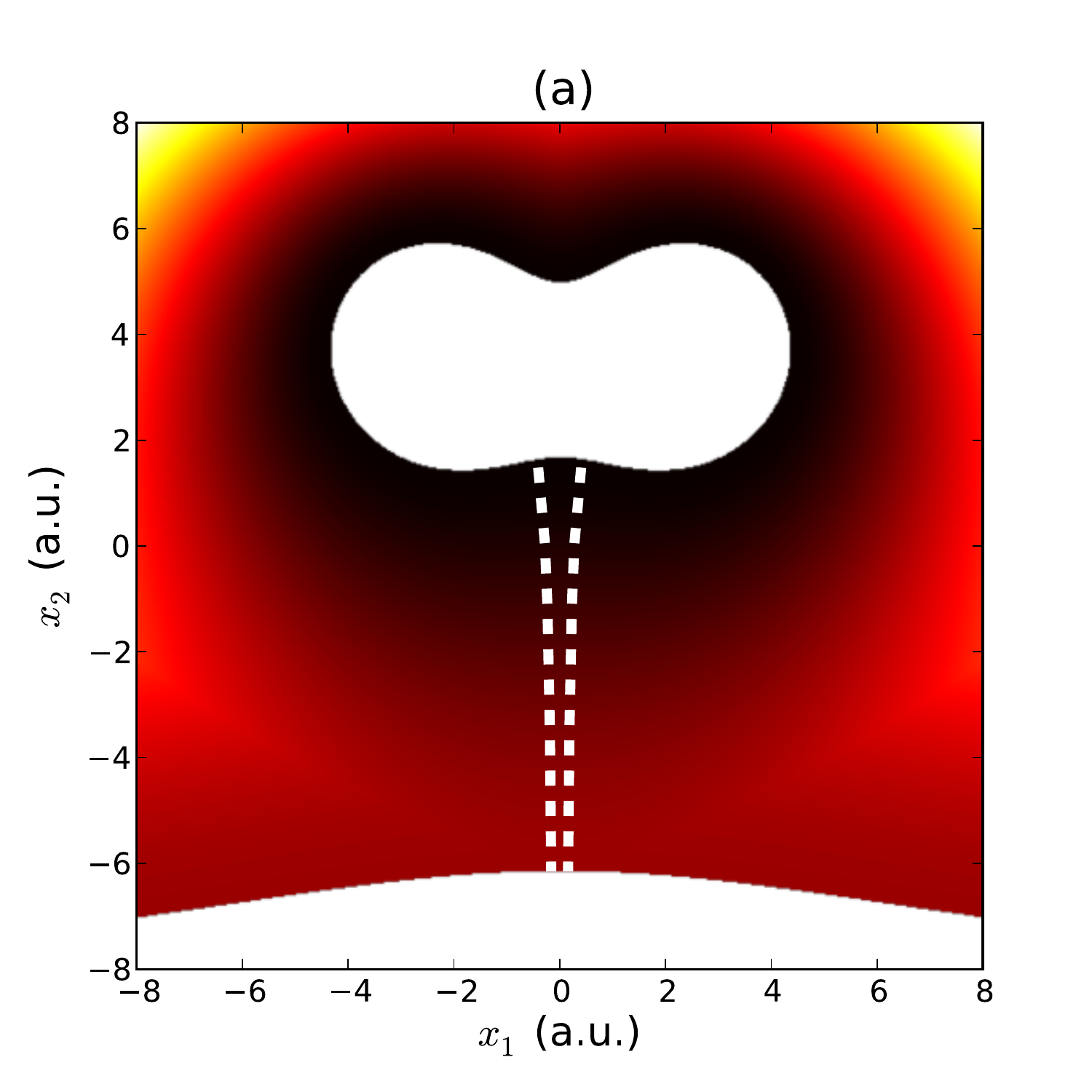}
	\includegraphics[scale=0.45]{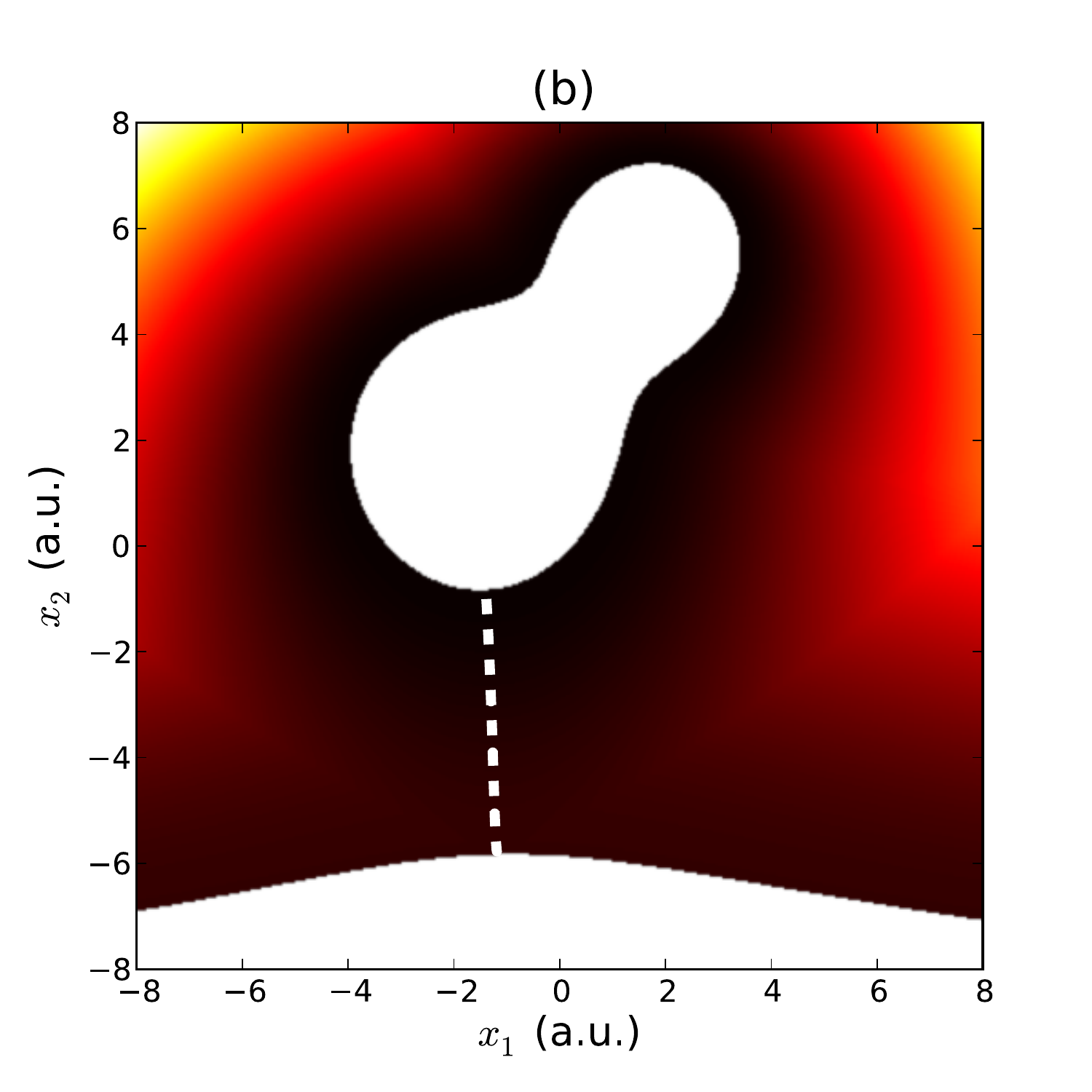}
	\includegraphics[scale=0.45]{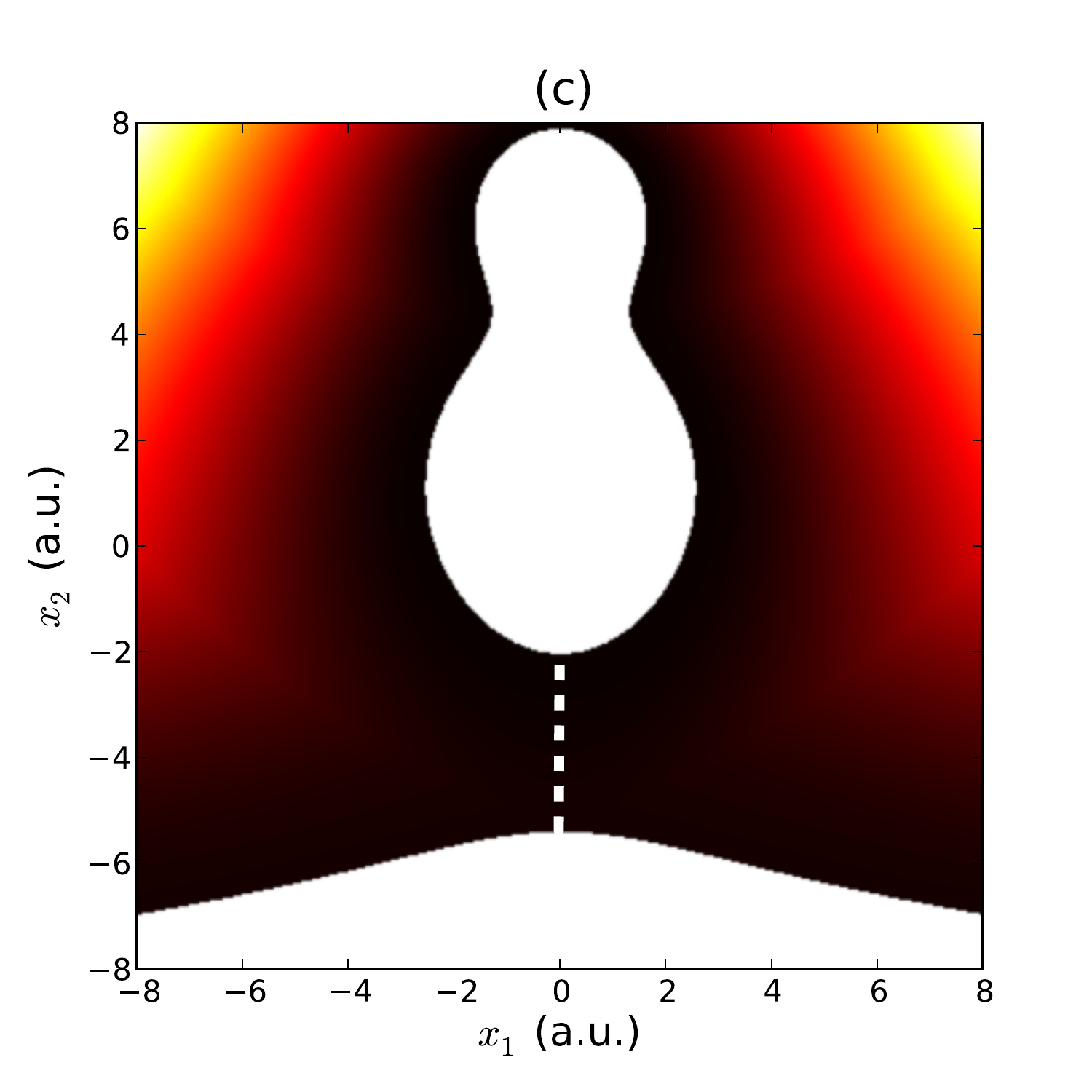}
	\includegraphics[scale=0.45]{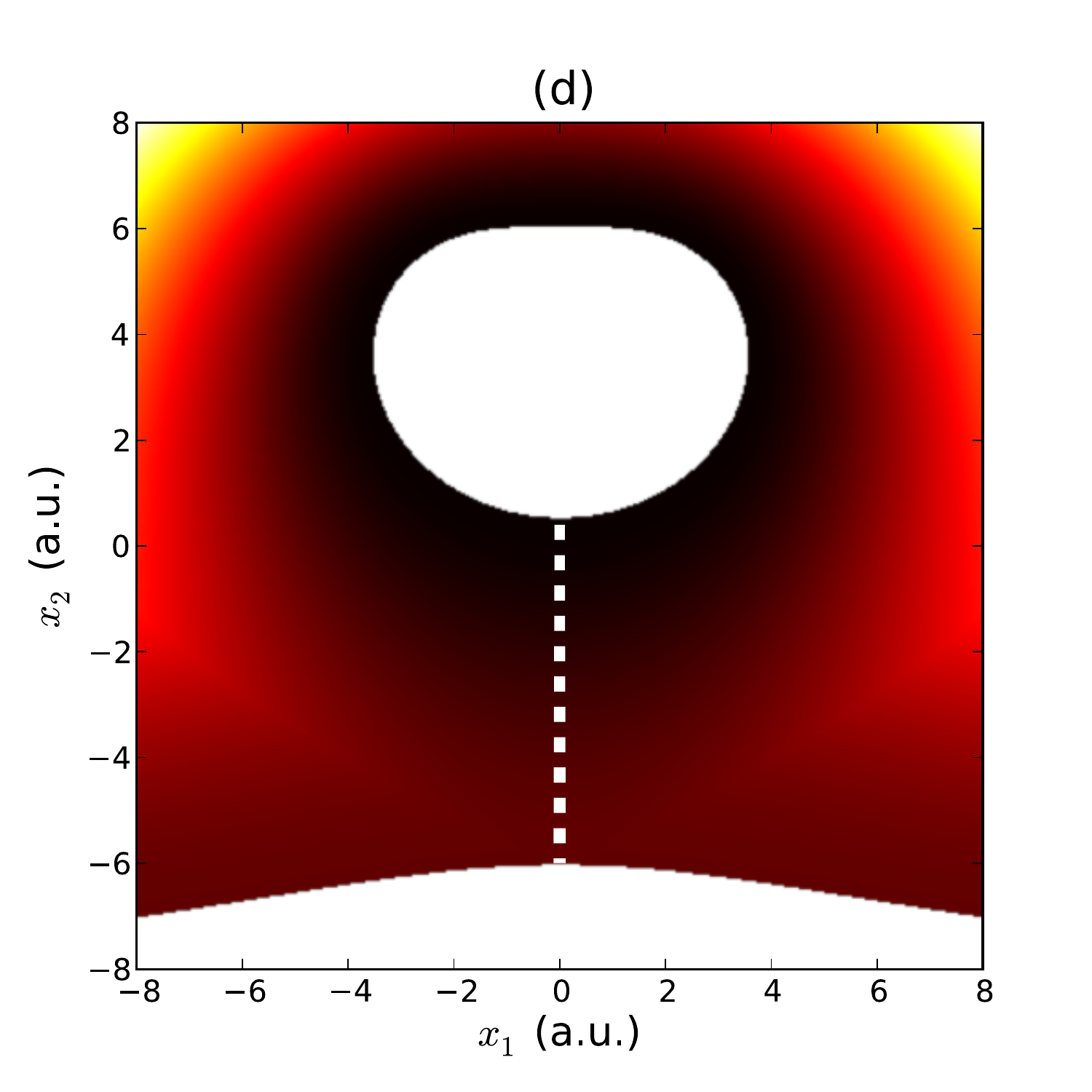}
\caption{(Color online) Viscosity solutions of eikonal equation (\ref{Hamiljacobi_AgmonDistance}) and leading tunnelling trajectories in the case of two identical overlapping long range potentials, which are given by Eq. (\ref{LongRangeDiatomic}) with $Z_1 = Z_2 =1$. White colour denotes the classically allowed regions. Dashed white lines are leading tunnelling trajectories. The solutions of the eikonal equation is represented by linear scale colour ramps from black (minimum) to bright color (maximum). (a) $\theta = 90^{\circ}$ and $R = 5$ a.u.; (b) $\theta = 45^{\circ}$ and $R = 5$ a.u.; (c) $\theta = 0^{\circ}$ and $R = 5$ a.u.; (d) $\theta = 90^{\circ}$ and $R=3$ a.u.}\label{FigOverlapingLongRange}
\end{center}
\end{figure*}

\begin{center}
\begin{figure}
	\includegraphics[scale=0.45]{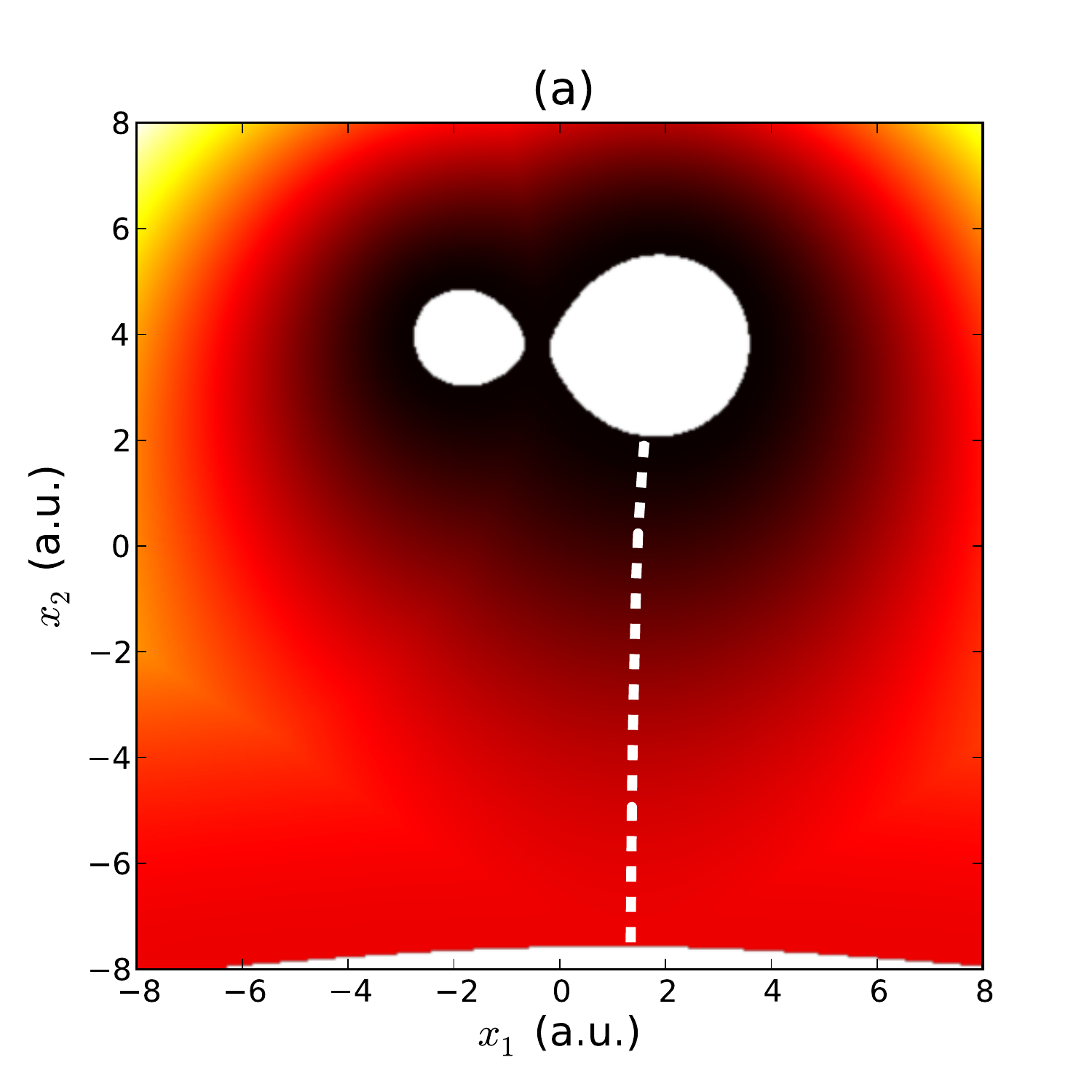}
	\includegraphics[scale=0.45]{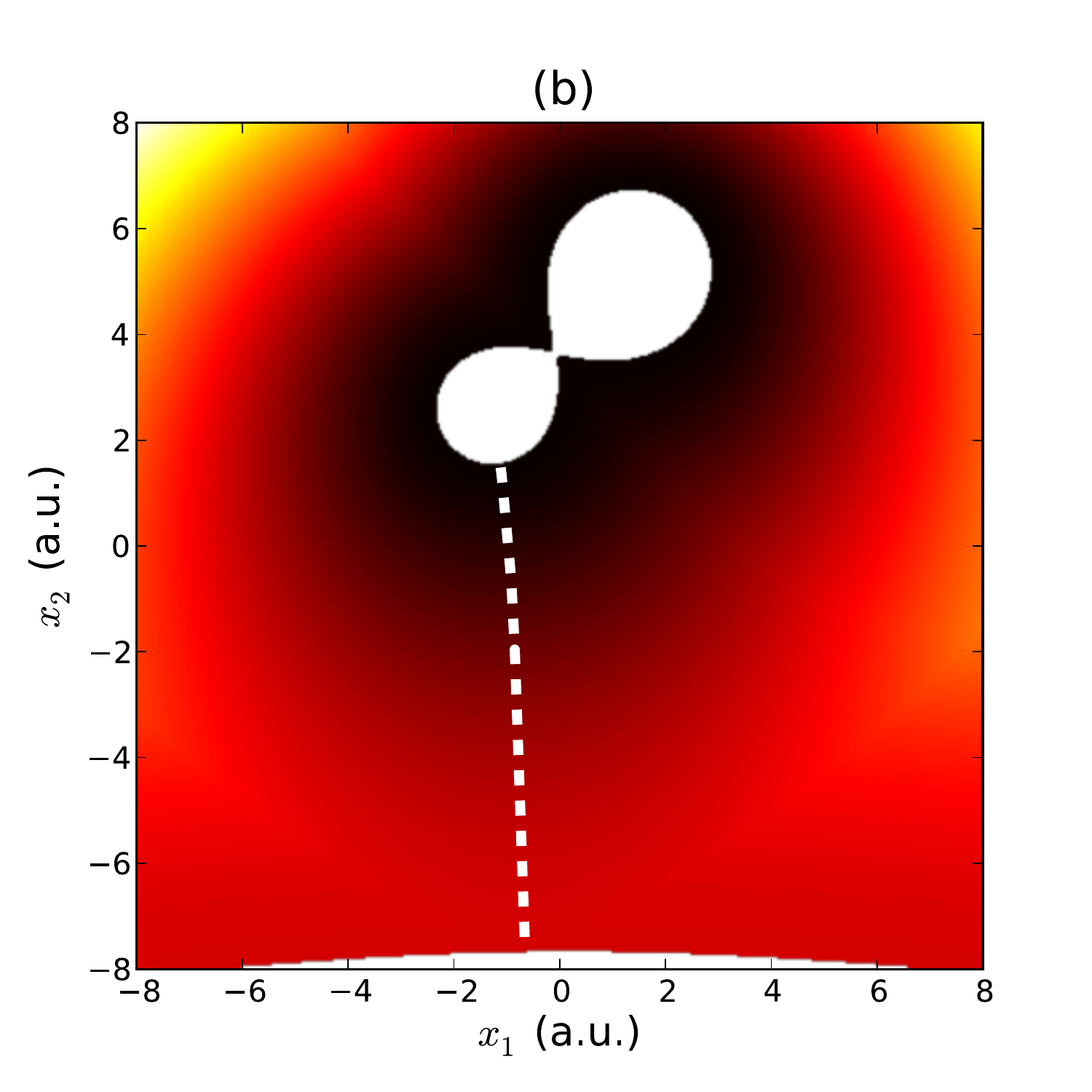}
	\includegraphics[scale=0.45]{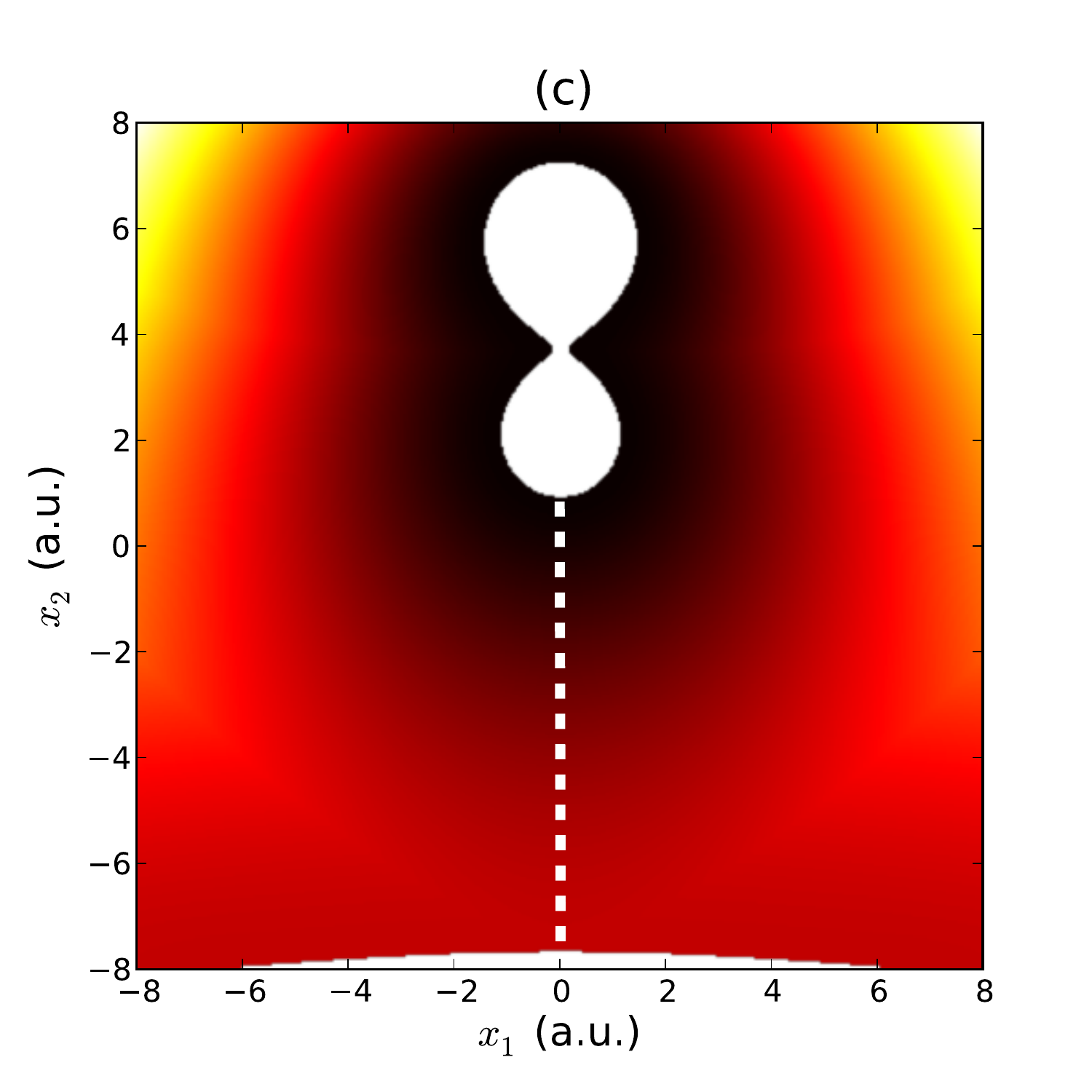}
\caption{(Color online) Viscosity solutions of eikonal equation (\ref{Hamiljacobi_AgmonDistance}) and leading tunnelling trajectories in the case of two non-identical long range potentials, which are given by Eq. (\ref{LongRangeDiatomic}) with $Z_1 = 0.4$, $Z_2 = 1$, and $R = 3.9$ a.u. White colour denotes the classically allowed regions. Dashed white lines are leading tunnelling trajectories. The solutions of the eikonal equation is represented by linear scale colour ramps from black (minimum) to bright color (maximum). (a) $\theta = 90^{\circ}$; (b) $\theta = 45^{\circ}$; (c) $\theta = 0^{\circ}$.}\label{NonIdenticalLongRange}
\end{figure}
\end{center}

\section{Conclusions and discussions}

Having introduced the leading tunnelling trajectory as an instanton path that gives the highest tunnelling probability, we have proven that leading tunnelling trajectories for multi-centre short range potentials are linear (Theorem \ref{theorem1}) and ``almost'' linear for multi-centre long range potentials (Theorem \ref{theorem2} and the rule of thumb from Sec. \ref{Sec5}).

The fact that the leading trajectories for long range potentials are not straight lines is of vital importance. As in the atomic case \cite{Perelomov1966a, Perelomov67a, Perelomov1967a, Popov1968, Popov2004, Popov2005,  Popruzhenko2008a, Popruzhenko2008b, Popruzhenko2009}, this deviation is crucial for a quantitative treatment \cite{Tong2002, Fabrikant2009, Murray2010, Bin2010, Fabrikant2010}, and sometimes even for a qualitative analysis, because it leads to the correct pre-exponential factor of ionization rates that describes the influence of the Coulomb field of nuclei. However, Theorem \ref{theorem2} as well as illustrations presented in Sec. \ref{Sec5} suggests that the exact shape of the trajectory can be obtained as a perturbation of the trajectory in the modified potential where the Coulomb potential is substituted by the finite range one. This is a part of the celebrated Perelomov-Popov-Terent'ev (PPT) approach \cite{Perelomov1966a, Perelomov67a, Perelomov1967a, Popov1968}, widely employed in the literature for analytical calculations of atomic Coulomb corrections. Nevertheless, the PPT method requires matching the quasiclassical wave function of an electron in the continuum with the bound (field-free) atomic wave function. This step is a stumbling block for generalization of the PPT approach to the molecular case (for the suggestion of a solution to such a problem see Refs. \cite{Murray2010, Murray2011}). Theorem \ref{theorem2} in fact offers a solution to the problem of matching. According to Theorem \ref{theorem2}, matching should be done on spherical surfaces of radii $\Delta_j$ [Eq. (\ref{DeltaDeff})] centred at the nuclei. This is an alternative technique to the method developed in Refs. \cite{Murray2010, Murray2011}; however, it is applicable only to well separated core electrons and may not be applicable to delocalized valence electrons. 

It has been suggested in Ref. \cite{Meckel2008a} that molecular photoionization in the tunnelling limit may act as a scanning tunnelling microscope (STM). Since rotating a molecule with respect to a field direction is analogous to moving the tip of an STM, then the observed angular-dependent ionization probability should provide information for a molecule similar to the position dependence of the tunnelling current in the STM. We point out that there is a resemblance between such a descriptive comparison and our results. As we have demonstrated (the rule of thumb and the backward propagation of the leading trajectory), the leading tunnelling trajectory starts at the atomic centre that is the closest to the barrier exit  (i.e., the outer turning surface); hence, the qualitative similarity of molecular tunnelling with the STM. Nevertheless, one must bear in mind that an electron cannot tunnel along a path because Heisenberg's uncertainty principle requires a wave packet with a finite lateral extension. It is demonstrated in Ref. \cite{Bracher1998} how the tunnelling current is drawn out of a bound state along the direction of the external field, spreading out also somewhat in the orthogonal direction, and giving rise to a tunnelling ``spot'' of approximately a Gaussian shape. Such wave packets are important for the interpretation of experiments (see, e.g., Refs. \cite{Murray2011, Meckel2008a}).

The demonstrated simplicity of the shapes of leading tunnelling
trajectories, in fact, may encourage future developments of analytical
models of molecular ionization. Nevertheless, the geometrical approach has
a fundamental limitation -- it does not account for effects of molecular
orbitals, and there is no a priori way of including these effects. In
spite of that, one may always attempt to introduce such corrections in a
heuristic manner, e.g., multiplying the geometrical rates by a Dyson orbital.

In the current paper, we modelled a molecule by a single-electron multi-centre potential, hence discarding effects of electron-electron correlations. However, the geometrical approach to tunnelling reviewed in Sec. \ref{Sec2} can account for these effects after an appropriate adaptation presented in Ref. \cite{Sigal1988}. Intuitively speaking, according to such a method, the leading tunneling trajectory of the system is selected such that the minimum number of electrons are displaced during tunnelling. More importantly, the fast marching method still can be utilized to obtain this leading tunneling trajectory. Since correlation dynamics of electrons plays an important role in molecular ionization leading to interesting novel effects \cite{Harumiya2002, Awasthi2006, Vanne2008, Vanne2009, Vanne2010, Walters2010, Hoff2010}, applications of the geometrical ideas developed in Ref. \cite{Sigal1988} to molecular ionization should be the aim of subsequent publications.

\acknowledgments

The authors thank anonymous referees for a number of important suggestions that have significantly improved the paper. 
We are in debt to Misha~Yu.~Ivanov, Michael~Spanner, Ryan~Murray, and Olga~Smirnova for fruitful discussions. 
D.I.B. acknowledges the Ontario Graduate Scholarship program for financial support. W.K.L. acknowledges support of NSERC discovery grants.

\appendix

\section{Upper bounds for matrix elements and transition amplitudes}

The current section contains simple derivations of the multi-dimensional generalization of the Landau rule for calculation of the quasi-classical matrix elements between bound states [Eq. (\ref{Appendix_A_LandauQuasiClassMatrElem})] as well as estimates of perturbation theory transition amplitudes [Eqs. (\ref{AppendixA_UpperBoundA2}) and (\ref{AppendixA_UpperBoundA1})].  Note, nevertheless, that we do not employ these results in the paper. The purpose of these derivations is to demonstrate methodologically the utility of the Agmon upper bounds for bound states reviewed in Sec. \ref{Sec2} as a prelude to the Agmon geometrical ideas used to describe tunnelling.

For the sake of simplicity, the argument ${\bf x}$ will be omitted in some equations below. Throughout this Appendix, we assume that the Agmon upper bounds \cite{Agmon1982} for bound states ($\psi_n$) are valid, i.e., $\forall \epsilon > 0$ $\exists c_n \equiv c_n(\epsilon)$, $0 < c_n < \infty$, such that 
\begin{eqnarray}\label{AppendixA_UppBoundAssump}
 |\psi_n | \leqslant c_n e^{-(1-\epsilon) \rho_n},
\end{eqnarray}
where $\rho_n = \rho_{E_n}$.

Let us choose an arbitrary $\epsilon > 0$. Employing the Schwartz inequality and assumption (\ref{AppendixA_UppBoundAssump}), we obtain
\begin{eqnarray}\label{AppendixA_Ineq}
&& \left| \int \psi_p^* V \psi_q d {\bf x} \right|^2  \nonumber\\
&&\qquad =  \left| \int e^{(1-\epsilon)(\rho_p + \rho_q)} \psi_p^* \psi_q e^{-(1-\epsilon)(\rho_p + \rho_q)} V d {\bf x}\right|^2 \nonumber\\
&&\qquad \leqslant B_{p, q}^2 \int \left| e^{(1-\epsilon)(\rho_p + \rho_q)} \psi_p^* \psi_q \right|^2  d {\bf x} \nonumber\\
&&\qquad \leqslant B_{p, q}^2 c_p^2(\epsilon') c_q^2(\epsilon')\int e^{-2(\epsilon -\epsilon')(\rho_p + \rho_q)} d{\bf x},
\end{eqnarray}
where $\epsilon > \epsilon' > 0$ and
\begin{eqnarray}
B_{p, q}^l = \int |V|^l e^{-l(1-\epsilon)(\rho_p + \rho_q)} d {\bf x}.
\end{eqnarray}
The integral $\int \exp[-2(\epsilon -\epsilon')(\rho_p + \rho_q)] d{\bf x}$ converges for all $p$ and $q$. Therefore, we have proven that $\forall \epsilon >0$ $\exists c = c(\epsilon)$, $0 < c < \infty$, such that
\begin{eqnarray}\label{Appendix_A_LandauQuasiClassMatrElem}
\left| \int \psi_p^* V \psi_q d {\bf x} \right|^2 \leqslant c B_{p, q}^2,
\end{eqnarray}
which is the same as Eq. (\ref{Inequality_LandauQuasiclassicalMatrixElem}).

Now let us study the problem of estimating of transition amplitudes defined by means of the time dependent perturbation theory. Hereinafter, we assume that a quantum system under scrutiny has no continuum spectrum, and we shall manipulate with all the series and integrals heuristically --  assuming that they all converge, or alternatively, assuming that they are over a finite range. We illustrate our idea by estimating the second order amplitude since generalization to higher orders is evident. 

The second order transition amplitude within the time dependent perturbation theory reads
\begin{eqnarray}
A^{(2)} &=& - \int_{t_i}^{t_f} dt \int^{t_f}_t dt' \int d{\bf x} d{\bf x}' \psi_{fin}^*({\bf x}') e^{-iE_{fin}(t_f-t')}  \nonumber\\
&& \times V({\bf x}') K({\bf x}' t'|{\bf x} t)V({\bf x})\psi_{in}({\bf x})e^{-iE_{in}(t-t_i)},
\end{eqnarray}
where all the $\psi$'s are eigenstates of the system and $K$ is the propagator, which can be written as
\begin{eqnarray}\nonumber
K({\bf x}' t'|{\bf x} t) = \sum_n \psi_n({\bf x}') \psi_n^* ({\bf x}) e^{-iE_n(t' - t)};
\end{eqnarray}
whence, 
$$
|K({\bf x}' t'|{\bf x} t)| \leqslant \sum_n |\psi_n({\bf x}') \psi_n ({\bf x}) |.
$$
Using such a simple estimate as well as inequality (\ref{AppendixA_UppBoundAssump}), we obtain
\begin{eqnarray}\label{AppendixA_UpperBoundA2}
\frac{\left| A^{(2)}\right|}{ (t_f - t_i)^2 } &\leqslant& \frac{c_{in}c_{fin}}2 \sum_n c_n^2 B_{fin, n}^1 B_{n, in}^1 \nonumber\\
&\leqslant& M \sum_n B_{fin, n}^1 B_{n, in}^1,
\end{eqnarray}
where $M \equiv c_{in}c_{fin} \max_n \left\{ c_n^2 \right\} /2$, $0 < M < \infty$.

However, there is no need to confine ourself to the case when the initial and final states are eigenstates. The same idea applies to the general case of the initial ($\phi_{in}$) and final ($\phi_{fin}$) states being represented as linear expansions in the basis of the bound eigenstates, 
\begin{eqnarray}
\phi_{in} = \sum_n \langle \psi_n \ket{\phi_{in}} \psi_n, \quad
\phi_{fin} = \sum_n \langle \psi_n \ket{\phi_{fin}} \psi_n.
\end{eqnarray}
Let us found an upper bound for the first order transition amplitude, which is as follows
\begin{eqnarray}
A^{(1)} &=& -i\int_{t_i}^{t_f} dt \int d{\bf x} \sum_{n, n'} \langle \phi_{fin} \ket{\psi_n} \psi_n^*({\bf x}) e^{-iE_n(t_f-t)} \nonumber\\
&& \times V({\bf x}) \langle \psi_{n'} \ket{\phi_{in}} \psi_{n'}({\bf x}) e^{-iE_{n'}(t-t_i)}.
\end{eqnarray}
Whence, we readily obtain
\begin{eqnarray}\label{AppendixA_UpperBoundA1}
\frac{\left|A^{(1)}\right|}{t_f-t_i} &\leqslant& \sum_{n, n'} c_n c_{n'} |\langle \phi_{fin} \ket{\psi_n} \langle \psi_{n'} \ket{\phi_{in}}| B_{n, n'}^1 \nonumber\\
&\leqslant& M \sum_{n,n'} B_{n, n'}^1,
\end{eqnarray}
where $M \equiv \max_{n, n'} \left\{ c_n c_{n'} |\langle \phi_{fin} \ket{\psi_n} \langle \psi_{n'} \ket{\phi_{in}}| \right\}$, $0 < M < \infty$.

\bibliography{TrajectMolecularTunnelling}
\end{document}